\documentclass[sigplan, screen, authorversion]{acmart}
\newbool{EXTENDED}
\booltrue{EXTENDED}

\startPage{1}

\setcopyright{acmlicensed}
\acmPrice{15.00}
\acmDOI{10.1145/3242744.3242747}
\acmYear{2018}
\copyrightyear{2018}
\acmISBN{978-1-4503-5835-4/18/09}
\acmConference[Haskell '18]{Proceedings of the 11th ACM SIGPLAN International Haskell Symposium}{September 27--28, 2018}{St. Louis, MO, USA}
\acmBooktitle{Proceedings of the 11th ACM SIGPLAN International Haskell Symposium (Haskell '18), September 27--28, 2018, St. Louis, MO, USA}

\makeatletter
\@ifundefined{lhs2tex.lhs2tex.sty.read}%
  {\@namedef{lhs2tex.lhs2tex.sty.read}{}%
   \newcommand\SkipToFmtEnd{}%
   \newcommand\EndFmtInput{}%
   \long\def\SkipToFmtEnd#1\EndFmtInput{}%
  }\SkipToFmtEnd

\newcommand\ReadOnlyOnce[1]{\@ifundefined{#1}{\@namedef{#1}{}}\SkipToFmtEnd}
\usepackage{amstext}
\usepackage{amssymb}
\usepackage{stmaryrd}
\DeclareFontFamily{OT1}{cmtex}{}
\DeclareFontShape{OT1}{cmtex}{m}{n}
  {<5><6><7><8>cmtex8
   <9>cmtex9
   <10><10.95><12><14.4><17.28><20.74><24.88>cmtex10}{}
\DeclareFontShape{OT1}{cmtex}{m}{it}
  {<-> ssub * cmtt/m/it}{}

\DeclareFontShape{OT1}{cmtt}{bx}{n}
  {<5><6><7><8>cmtt8
   <9>cmbtt9
   <10><10.95><12><14.4><17.28><20.74><24.88>cmbtt10}{}
\DeclareFontShape{OT1}{cmtex}{bx}{n}
  {<-> ssub * cmtt/bx/n}{}

\newcommand{\Conid}[1]{\mathit{#1}}
\newcommand{\Varid}[1]{\mathit{#1}}
\newcommand{\anonymous}{\kern0.06em \vbox{\hrule\@width.5em}}
\newcommand{\plus}{\mathbin{+\!\!\!+}}

\renewcommand{\leq}{\leqslant}

\usepackage{polytable}

\@ifundefined{mathindent}%
  {\newdimen\mathindent\mathindent\leftmargini}%
  {}%

\def\resethooks{%
  \global\let\SaveRestoreHook\empty
  \global\let\ColumnHook\empty}
\newcommand*{\savecolumns}[1][default]%
  {\g@addto@macro\SaveRestoreHook{\savecolumns[#1]}}
\newcommand*{\restorecolumns}[1][default]%
  {\g@addto@macro\SaveRestoreHook{\restorecolumns[#1]}}
\newcommand*{\aligncolumn}[2]%
  {\g@addto@macro\ColumnHook{\column{#1}{#2}}}

\resethooks

\newcommand{\onelinecommentchars}{\quad-{}- }
\newcommand{\commentbeginchars}{\enskip\{-}
\newcommand{\commentendchars}{-\}\enskip}

\newcommand{\visiblecomments}{%
  \let\onelinecomment=\onelinecommentchars
  \let\commentbegin=\commentbeginchars
  \let\commentend=\commentendchars}

\newcommand{\invisiblecomments}{%
  \let\onelinecomment=\empty
  \let\commentbegin=\empty
  \let\commentend=\empty}

\visiblecomments

\newlength{\blanklineskip}
\newcommand{\hsindent}[1]{\quad}
\let\hspre\empty
\let\hspost\empty

\EndFmtInput
\makeatother
\ReadOnlyOnce{polycode.fmt}%
\makeatletter

\newcommand{\hsnewpar}[1]%
  {{\parskip=0pt\parindent=0pt\par\vskip #1\noindent}}

\newcommand{\hscodestyle}{}

\newcommand{\sethscode}[1]%
  {\expandafter\let\expandafter\hscode\csname #1\endcsname
   \expandafter\let\expandafter\endhscode\csname end#1\endcsname}

  {\par\noindent
   \advance\leftskip\mathindent
   \hscodestyle
   \let\\=\@normalcr
   \let\hspre\(\let\hspost\)%
   \pboxed}%
  {\endpboxed\)%
   \par\noindent
   \ignorespacesafterend}

  {\hsnewpar\abovedisplayskip
   \advance\leftskip\mathindent
   \hscodestyle
   \let\hspre\(\let\hspost\)%
   \pboxed}%
  {\endpboxed%
   \hsnewpar\belowdisplayskip
   \ignorespacesafterend}

  {\hsnewpar\abovedisplayskip
   \advance\leftskip\mathindent
   \hscodestyle
   \let\\=\@normalcr
   \(\pboxed}%
  {\endpboxed\)%
   \hsnewpar\belowdisplayskip
   \ignorespacesafterend}

\newcommand{\plainhs}{\sethscode{plainhscode}}

\plainhs

  {\hsnewpar\abovedisplayskip
   \advance\leftskip\mathindent
   \hscodestyle
   \let\\=\@normalcr
   \(\parray}%
  {\endparray\)%
   \hsnewpar\belowdisplayskip
   \ignorespacesafterend}

  {\parray}{\endparray}

  {\(\parray}{\endparray\)}

\def\codeframewidth{\arrayrulewidth}
\RequirePackage{calc}

  {\parskip=\abovedisplayskip\par\noindent
   \hscodestyle
   \arrayrulewidth=\codeframewidth
   \tabular{@{}|p{\linewidth-2\arraycolsep-2\arrayrulewidth-2pt}|@{}}%
   \hline\framedhslinecorrect\\{-1.5ex}%
   \let\endoflinesave=\\
   \let\\=\@normalcr
   \(\pboxed}%
  {\endpboxed\)%
   \framedhslinecorrect\endoflinesave{.5ex}\hline
   \endtabular
   \parskip=\belowdisplayskip\par\noindent
   \ignorespacesafterend}

\newcommand{\framedhslinecorrect}[2]%
  {#1[#2]}

  {\(\def\column##1##2{}%
   \let\>\undefined\let\<\undefined\let\\\undefined
   \newcommand\>[1][]{}\newcommand\<[1][]{}\newcommand\\[1][]{}%
   \def\fromto##1##2##3{##3}%
   }{\) }%

  {\let\orighscode=\hscode
   \let\origendhscode=\endhscode
   \def\endhscode{\def\hscode{\endgroup\def\@currenvir{hscode}\\}\begingroup}
   \orighscode\def\hscode{\endgroup\def\@currenvir{hscode}}}%
  {\origendhscode
   \global\let\hscode=\orighscode
   \global\let\endhscode=\origendhscode}%

\makeatother
\EndFmtInput
\usepackage{hyperref}
\usepackage{framed}
\usepackage{paralist}
\usepackage{calc}
\usepackage{graphicx}
\usepackage{thmtools, thm-restate}
\usepackage{mathpartir}
\usepackage{pgfplots}
\usepackage{booktabs}
\usepackage{wrapfig}
\usepackage{stackrel}
\usepackage{xspace}
\usepackage{multicol}
\usepackage{makecell}
\usepackage{balance}

\allowdisplaybreaks

\usepackage{tikz}
\usepackage{tikz-qtree,tikz-qtree-compat}
\usepackage{tikz-cd}
\usetikzlibrary{positioning,shapes}
\usetikzlibrary{fit}

\tikzstyle{high} =
   [rectangle, draw=red, rounded corners, text centered]

\tikzstyle{low} = 
  [rectangle, draw=blue, rounded corners, text centered]
\pgfplotsset{compat=1.12}
\usepgfplotslibrary{statistics}

\newcommand{\nop}[1]{\!\!}

\newcommand{\dragen}{\dragenlogo\xspace}
\newcommand{\dragenfull}{\dragenlogofull\xspace}
\newcommand\quickcheck{\emph{QuickCheck}\xspace}
\newcommand\derive{\emph{derive}\xspace}
\newcommand\megadeth{\emph{MegaDeTH}\xspace}
\newcommand\quickfuzz{\emph{QuickFuzz}\xspace}
\newcommand\feat{\emph{Feat}\xspace}
\newcommand\godeltest{\emph{G\"odelTest}\xspace}

\definecolor{dark-gray}{RGB}{50,50,50}
\definecolor{gray}{RGB}{84,84,84}
\definecolor{bgcol}{RGB}{255,255,240}
\definecolor{dark-green}{RGB}{0,100,0}

\declaretheorem[name=Definition, style=definition, numbered=no]{helpers-def}
\declaretheorem[name=Proposition, style=lemma, numberlike=subsection]{mC}
\declaretheorem[name=Proposition, style=lemma, numberlike=subsection]{mT}
\declaretheorem[name=Proposition, style=lemma, numberlike=subsection]{mC-vs-mT}
\declaretheorem[name=Theorem, style=theorem]{types-matrix}

\newenvironment{CompactItemize}%
  {\begin{list}{$\; \; \; \; \; \; \; \ \  \blacktriangleright$}%
   {\leftmargin=0pt \itemsep=2pt \topsep=5pt
     \parsep=0pt \partopsep=0pt}}%
  {\end{list}}

\newlength\htG\newlength\dpg
\protected\def\dragenlogo{\settoheight{\htG}{G}\settodepth{\dpg}{g}%
  \raisebox{-0.0\dpg}{\includegraphics[height=0.6\htG+\dpg]{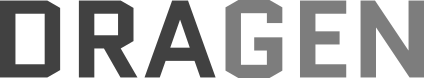}}}
\protected\def\dragenlogofull{\settoheight{\htG}{G}\settodepth{\dpg}{g}%
  \raisebox{-0.0\dpg}{\includegraphics[height=0.6\htG+\dpg]{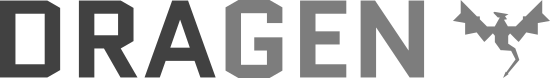}}}

\graphicspath{{plots/}}

\begin{document}


\title[Branching Processes for QuickCheck Generators]{Branching Processes for QuickCheck Generators}


\ifbool{EXTENDED}{
  \subtitle{(extended version)}
}





\author{Agust\'in Mista}
\affiliation{
  \institution{Universidad Nacional de Rosario} 
  \city{Rosario}
  \country{Argentina}                             
}
\email{amista@dcc.fceia.unr.edu.ar}             

\author{Alejandro Russo}
\affiliation{
  \institution{Chalmers University of Technology}
  \city{Gothenburg}
  \country{Sweden}
}
\email{russo@chalmers.se}

\author{John Hughes}
\affiliation{
  \institution{Chalmers University of Technology}
  \city{Gothenburg}
  \country{Sweden}
}
\email{rjmh@chalmers.se}



\begin{abstract}
%
%
%
%
%
%
In \quickcheck (or, more generally, random testing), it is challenging to
control random data generators' distributions---specially when it comes to
\emph{user-defined algebraic data types} (ADT).
%
%
In this paper, we adapt results from an area of mathematics known as
\emph{branching processes}, and show how they help to analytically predict (at
compile-time) the expected number of generated constructors, even in the
presence of mutually recursive or composite ADTs.
%
%
%
Using our probabilistic formulas, we design heuristics capable of automatically
adjusting probabilities in order to synthesize generators which distributions
are aligned with users' demands.
%
%
%
We provide a Haskell implementation of our mechanism in a tool called \dragen
and perform case studies with real-world applications.
When generating random values, our synthesized \quickcheck generators show
improvements in code coverage when compared with those automatically derived by
state-of-the-art tools.
\end{abstract}

\ccsdesc[500]{Software and its engineering~Software testing and debugging}


\keywords{Branching process, QuickCheck, Testing, Haskell}

\maketitle


\section{Introduction}
Random property-based testing is an increasingly popular approach to finding
bugs \cite{HughesNSA16,ArtsHNS15,HughesPAN16}.
In the Haskell community, \quickcheck \cite{ClaessenH00} is the dominant tool of
this sort.
%
%
%
%
%
%
%
\quickcheck requires developers to specify \emph{testing properties} describing
the expected software behavior.
%
%
%
Then, it generates a large number of random \emph{test cases} and reports those
violating the testing properties.
\quickcheck generates random data by employing \emph{random test data
generators} or \quickcheck generators for short.
The generation of test cases is guided by the \emph{types} involved in the
testing properties.
It defines default generators for many built-in types like booleans, integers,
and lists.
However, when it comes to user-defined ADTs, developers are usually required to
specify the generation process.
%
%
The difficulty is, however, that it might become intricate to define generators
so that they result in a suitable distribution or enforce data invariants.
The state-of-the-art tools to derive generators for user-defined ADTs can be
classified based on the automation level as well as the sort of invariants
enforced at the data generation phase.
%
%
%
%
\quickcheck and \emph{SmallCheck} \cite{RuncimanNL08} (a tool for writing
generators that synthesize small test cases) use type-driven generators written
by developers.
As a result, generated random values are well-typed and preserve the structure
described by the ADT.
Rather than manually writing generators, libraries \derive \cite{mitchell2007}
and \megadeth \cite{GriecoCB16, grieco2017} automatically synthesize generators
for a given user-defined ADT.
The library \derive provides no guarantees that the generation process
terminates, while \megadeth pays almost no attention to the distribution of
values.
In contrast, \emph{Feat} \cite{DuregardJW12} provides a mechanism to uniformly
sample values from a given ADT.
It enumerates all the possible values of a given ADT so that sampling uniformly
from ADTs becomes sampling uniformly from the set of natural numbers.
\emph{Feat}'s authors subsequently extend their approach to \emph{uniformly}
generate values constrained by user-defined predicates \cite{ClaessenDP14}.
%
Lastly, \emph{Luck} is a domain specific language for manually writing
\quickcheck properties in tandem with generators so that it becomes possible to
finely control the distribution of generated values \cite{LampropoulosGHH17}.

%
In this work, we consider the scenario where developers are not fully aware of
the properties and invariants that input data must fulfill.
This constitutes a valid assumption for \emph{penetration testing}
\cite{pentest}, where testers often apply fuzzers in an attempt to make programs
crash---an anomaly which might lead to a vulnerability.
We believe that, in contrast, if users can recognize specific properties of
their systems then it is preferable to spend time writing specialized generators
for that purpose (e.g., by using \emph{Luck}) instead of considering
automatically derived ones.

Our realization is that \emph{branching processes} \cite{gw1875}, a relatively
simple stochastic model conceived to study the evolution of populations, can be
applied to predict the generation distribution of ADTs' constructors in a simple
and automatable manner.
%
%
To the best of our knowledge, this stochastic model has not yet been applied to
this field, and we believe it may be a promising foundation to develop future
extensions.
The contributions of this paper can be outlined as follows:
\vspace{-5pt}
\begin{CompactItemize}
\item We provide a mathematical foundation which helps to analytically
  characterize the distribution of constructors in derived \quickcheck
  generators for ADTs.
\item We show how to use type reification to simplify our prediction process and
  extend our model to mutually recursive and composite types.
\item We design (compile-time) heuristics that automatically search for
  probability parameters so that distributions of constructors can be adjusted
  to what developers might want.
  %
\item We provide an implementation of our ideas in the form of a Haskell
  library\footnote{Available at \url{https://bitbucket.org/agustinmista/dragen}}
  called \dragenfull (the Danish word for \emph{dragon}, here standing for
  \emph{Derivation of RAndom GENerators}).
\item We evaluate our tool by generating inputs for real-world programs, where
  it manages to obtain significantly more code coverage than those random inputs
  generated by \megadeth's generators.
\end{CompactItemize}

\vspace{-5pt}
Overall, our work addresses a timely problem with a neat mathematical insight
that is backed by a complete implementation and experience on third-party
examples.

%







\section{Background}
\label{sec:QC}

In this section, we briefly illustrate how \quickcheck random generators work.
%
%
%
%
%
We consider the following implementation of binary trees:
%
%
%
\begin{hscode}\SaveRestoreHook
\column{B}{@{}>{\hspre}l<{\hspost}@{}}%
\column{E}{@{}>{\hspre}l<{\hspost}@{}}%
\>[B]{}\mathbf{data}\;\Conid{Tree}\mathrel{=}Leaf_A\mid Leaf_B\mid Leaf_C\mid \Conid{Node}\;\Conid{Tree}\;\Conid{Tree}{}\<[E]%
\ColumnHook
\end{hscode}\resethooks

\begin{figure}[b] 
\vspace{-10pt}
\begin{framed}
\vspace{-10pt}
\begin{hscode}\SaveRestoreHook
\column{B}{@{}>{\hspre}l<{\hspost}@{}}%
\column{3}{@{}>{\hspre}l<{\hspost}@{}}%
\column{22}{@{}>{\hspre}c<{\hspost}@{}}%
\column{22E}{@{}l@{}}%
\column{25}{@{}>{\hspre}l<{\hspost}@{}}%
\column{E}{@{}>{\hspre}l<{\hspost}@{}}%
\>[B]{}\mathbf{instance}\;\Conid{Arbitrary}\;\Conid{Tree}\;\mathbf{where}{}\<[E]%
\\
\>[B]{}\hsindent{3}{}\<[3]%
\>[3]{}\Varid{arbitrary}\mathrel{=}\Varid{oneof}\;{}\<[22]%
\>[22]{}[\mskip1.5mu {}\<[22E]%
\>[25]{}\Varid{pure}\;Leaf_A,\Varid{pure}\;Leaf_B,\Varid{pure}\;Leaf_C{}\<[E]%
\\
\>[22]{},{}\<[22E]%
\>[25]{}\Conid{Node}\mathop{\langle \texttt{\$} \rangle}\Varid{arbitrary}\mathop{\langle \ast \rangle}\Varid{arbitrary}\mskip1.5mu]{}\<[E]%
\ColumnHook
\end{hscode}\resethooks
\vspace{-15pt}
\caption{\label{gen:derive}Random generator for \ensuremath{\Conid{Tree}}.}
\vspace{-5pt}
\end{framed}
\end{figure}



\noindent
In order to help developers write generators, \quickcheck defines the
\ensuremath{\Conid{Arbitrary}} type-class with the overloaded symbol \ensuremath{\Varid{arbitrary}\mathbin{::}\Conid{Gen}\;\Varid{a}}, which
denotes a monadic generator for values of type \ensuremath{\Varid{a}}.
%
%
%
%
%
%
%
%
%
Then, to generate random trees, we need to provide an instance of the
\ensuremath{\Conid{Arbitrary}} type-class for the type \ensuremath{\Conid{Tree}}.
%
%
%
Figure \ref{gen:derive} shows a possible implementation.
%
%
%
At the top level, this generator simply uses \quickcheck's primitive \ensuremath{\Varid{oneof}\mathbin{::}[\mskip1.5mu \Conid{Gen}\;\Varid{a}\mskip1.5mu]\to \Conid{Gen}\;\Varid{a}} to pick a generator from a list of generators with uniform
probability.
This list consists of a random generator for each possible choice of data
constructor of \ensuremath{\Conid{Tree}}.
%
We use \emph{applicative style} \cite{Mcbride2008} to describe each one of them
idiomatically.
So, \ensuremath{\Varid{pure}\;Leaf_A} is a generator that always generates \ensuremath{Leaf_A}s, while \ensuremath{\Conid{Node}\mathop{\langle \texttt{\$} \rangle}\Varid{arbitrary}\mathop{\langle \ast \rangle}\Varid{arbitrary}} 
is a generator that always generates \ensuremath{\Conid{Node}} constructors, ``filling'' its
arguments by calling \ensuremath{\Varid{arbitrary}} recursively on each of them.
\looseness=-1

Although it might seem easy, writing random generators becomes cumbersome very
quickly.
Particularly, if we want to write a random generator for a user-defined ADT \ensuremath{\Conid{T}},
it is also necessary to provide random generators for every user-defined ADT
inside of \ensuremath{\Conid{T}} as well!
%
%
%
%
What remains of this section is focused on explaining the state-of-the-art
techniques used to \emph{automatically} derive generators for user-defined ADTs
via type-driven approaches.

\subsection{Library \derive}
%
%
%
%
%
The simplest way to automatically derive a generator for a given ADT is the one
implemented by the Haskell library \derive \cite{mitchell2007}.
%
%
%
%
This library uses Template Haskell \cite{SheardJ02} to automatically synthesize
a generator for the data type \ensuremath{\Conid{Tree}} semantically equivalent to the one
presented in Figure \ref{gen:derive}.
%
%
%

While the library \derive is a big improvement for the testing process,
%
%
%
its implementation has a serious shortcoming when dealing with recursively
defined data types: in many cases, there is a non-zero probability of generating
a recursive type constructor every time a recursive type constructor gets
generated, which can lead to infinite generation loops.
A detailed example of this phenomenon%
\ifbool{EXTENDED}{
  is given in Appendix \ref{app:derive}.
}{
  is presented in the supplementary material \cite{extended}.
}
%
%
%
%
%
%
%
%
%
%
In this work, we only focus on derivation tools which accomplish terminating
behavior, since we consider this an essential component of well-behaved
generators.
\looseness=-1

\subsection{\megadeth}
%
%


%

\begin{figure}[b] 
\vspace{-10pt}
\begin{framed}
\vspace{-10pt}
\begin{hscode}\SaveRestoreHook
\column{B}{@{}>{\hspre}l<{\hspost}@{}}%
\column{3}{@{}>{\hspre}l<{\hspost}@{}}%
\column{5}{@{}>{\hspre}l<{\hspost}@{}}%
\column{7}{@{}>{\hspre}l<{\hspost}@{}}%
\column{23}{@{}>{\hspre}l<{\hspost}@{}}%
\column{42}{@{}>{\hspre}l<{\hspost}@{}}%
\column{E}{@{}>{\hspre}l<{\hspost}@{}}%
\>[B]{}\mathbf{instance}\;\Conid{Arbitrary}\;\Conid{Tree}\;\mathbf{where}{}\<[E]%
\\
\>[B]{}\hsindent{3}{}\<[3]%
\>[3]{}\Varid{arbitrary}\mathrel{=}\Varid{sized}\;\Varid{gen}\;\mathbf{where}{}\<[E]%
\\
\>[3]{}\hsindent{2}{}\<[5]%
\>[5]{}\Varid{gen}\;\mathrm{0}\mathrel{=}\Varid{oneof}{}\<[E]%
\\
\>[5]{}\hsindent{2}{}\<[7]%
\>[7]{}[\mskip1.5mu \Varid{pure}\;Leaf_A,\Varid{pure}\;Leaf_B,\Varid{pure}\;Leaf_C\mskip1.5mu]{}\<[E]%
\\
\>[3]{}\hsindent{2}{}\<[5]%
\>[5]{}\Varid{gen}\;\Varid{n}\mathrel{=}\Varid{oneof}{}\<[E]%
\\
\>[5]{}\hsindent{2}{}\<[7]%
\>[7]{}[\mskip1.5mu \Varid{pure}\;Leaf_A,\Varid{pure}\;Leaf_B,\Varid{pure}\;Leaf_C{}\<[E]%
\\
\>[5]{}\hsindent{2}{}\<[7]%
\>[7]{},\Conid{Node}\mathop{\langle \texttt{\$} \rangle}\Varid{gen}\;{}\<[23]%
\>[23]{}(\Varid{div}\;\Varid{n}\;\mathrm{2})\mathop{\langle \ast \rangle}\Varid{gen}\;{}\<[42]%
\>[42]{}(\Varid{div}\;\Varid{n}\;\mathrm{2})\mskip1.5mu]{}\<[E]%
\ColumnHook
\end{hscode}\resethooks
\vspace{-15pt}
\caption{\label{fig:mega} \megadeth generator for \ensuremath{\Conid{Tree}}.}
\vspace{-5pt}
\end{framed}
\end{figure}
%
%
The second approach we will discuss is the one taken by \megadeth, a
meta-programming tool used intensively by \quickfuzz
\cite{GriecoCB16, grieco2017}.
%
%
%
Firstly, \megadeth derives random generators for ADTs as well as all of its
nested types---a useful feature not supported by \derive.
%
%
%
Secondly, \megadeth avoids potentially infinite generation loops by setting an
upper bound to the random generation recursive depth.

Figure \ref{fig:mega} shows a simplified (but semantically equivalent) version
of the random generator for \ensuremath{\Conid{Tree}} derived by \megadeth.
%
%
This generator uses \quickcheck's function \ensuremath{\Varid{sized}\mathbin{::}(\Conid{Int}\to \Conid{Gen}\;\Varid{a})\to \Conid{Gen}\;\Varid{a}} to
build a random generator based on a function (of type \ensuremath{\Conid{Int}\to \Conid{Gen}\;\Varid{a}}) that
limits the possible recursive calls performed when creating random values.
%
%
The integer passed to \ensuremath{\Varid{sized}}'s argument is called the \emph{generation size}.
%
%
When the generation size is zero (see definition \ensuremath{\Varid{gen}\;\mathrm{0}}), the generator only
chooses between the \ensuremath{\Conid{Tree}}'s terminal constructors---thus ending the generation
process.
If the generation size is strictly positive, it is free to randomly generate any
\ensuremath{\Conid{Tree}} constructor (see definition \ensuremath{\Varid{gen}\;\Varid{n}}).
When it chooses to generate a recursive constructor, it reduces the generation
size for its subsequent recursive calls by a factor that depends on the number
of recursive arguments this constructor has (\ensuremath{\Varid{div}\;\Varid{n}\;\mathrm{2}}).
In this way, \megadeth ensures that all generated values are finite.
%
%
%
%

Although \megadeth generators always terminate, they have a major practical
drawback: in our example, the use of \ensuremath{\Varid{oneof}} to uniformly decide the next
constructor to be generated produces a generator that generates leaves
approximately three quarters of the time (note this also applies to the
generator obtained with \derive from Figure \ref{gen:derive}).
%
%
This entails a distribution of constructors heavily concentrated on leaves, with
a very small number of complex values with nested nodes, regardless how large
the chosen generation size is---see Figure \ref{fig:tree_megadeth_feat} (left).
\looseness=-1

\subsection{\feat}
The last approach we discuss is \feat \cite{DuregardJW12}.
This tool determines the distribution of generated values in a completely
different way:
%
it uses uniform generation based on an \emph{exhaustive enumeration of all the
  possible values of the ADTs being considered}.
\feat automatically establishes a bijection between all the possible values of a
given type \ensuremath{\Conid{T}}, and a finite prefix of the natural numbers.
%
%
%
%
%
%
%
%
Then, it guarantees a \emph{uniform generation over the complete space of values
  of a given data type \ensuremath{\Conid{T}}} up to a certain size.%
\footnote{We avoid including any source code generated by \feat, since it works
  by synthetizing \ensuremath{\Conid{Enumerable}} type-class instances instead of \ensuremath{\Conid{Arbitrary}} ones.
  Such instances give no insight into how the derived random generators work.
  %
}
%
%
However, the distribution of size, given by the number of constructors in the
generated values, is highly dependent on the structure of the data type being
considered.

Figure \ref{fig:tree_megadeth_feat} (right) shows the overall distribution shape
of a \quickcheck generator derived using \feat for \ensuremath{\Conid{Tree}} using a generation
size of 400, i.e., generating values of up to 400 constructors.%
\footnote{
  %
  We choose to use this generation size here since it helps us to compare
  \megadeth and \feat with the results of our tool in Section
  \ref{sec:casestudies}.
  %
}
%
%
Notice that all the generated values are close to the maximum size!
%
%
%
This phenomenon follows from the exponential growth in the number of possible
\ensuremath{\Conid{Tree}}s of \ensuremath{\Varid{n}} constructors as we increase \ensuremath{\Varid{n}}.
%
%
In other words, the space of \ensuremath{\Conid{Tree}}s up to 400 constructors is composed to a
large extent of values with around 400 constructors, and (proportionally) very
few with a smaller number of constructors.
Hence, a generation process based on uniform generation of a natural number
(which thus ignores the structure of the type being generated) is biased very
strongly towards values made up of a large number of constructors.
In our tests, no \ensuremath{\Conid{Tree}} with less than 390 constructors was ever generated.
In practice, this problem can be partially solved by using a variety of
generation sizes in order to get more diversity in the generated values.
However, to decide which generation sizes are the best choices is not a trivial
task either.
As consequence, in this work we consider only the case of fixed-size random
generation.

\begin{figure}[t] 
  \includegraphics[width=\columnwidth]{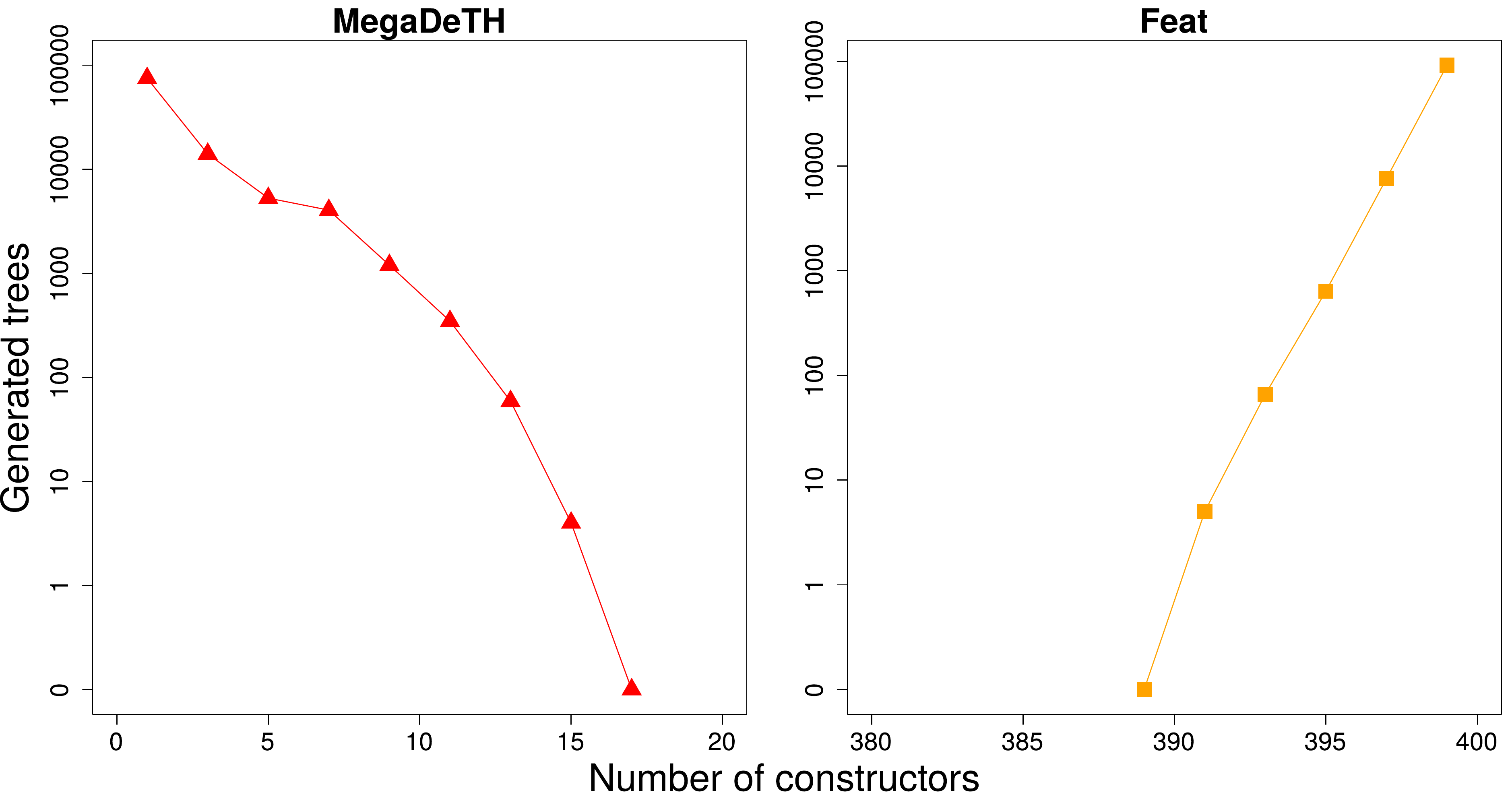}
  \caption{Size distribution of 100000 randomly generated \ensuremath{\Conid{Tree}} values using
    \megadeth ({\color{red} $\blacktriangle$}) with generation size 10, and
    \feat ({\color[RGB]{255, 163, 0} {\tiny $\blacksquare$}}) with generation
    size 400.}
  \label{fig:tree_megadeth_feat}
  \vspace{-5pt}
\end{figure}

%
As we have shown, by using both \megadeth and \feat, the user is tied to the
fixed generation distribution that each tool produces, which tends to be highly
dependent on the particular data type under consideration on each case.
Instead, this work aims to provide a \emph{theoretical framework able to predict
  and later tune the distributions of automatically derived generators}, giving
the user a more flexible testing environment, while keeping it as automated as
possible.

\section{Simple-Type Branching Processes}
\label{sec:bp}
Galton-Watson Branching processes (or branching processes for short) are a
particular case of Markov processes that model the growth and extinction of
populations.
%
%
%
Originally conceived to study the extinction of family names in the Victorian
era, this formalism has been successfully applied to a wide range of research
areas in biology and physics---see the textbook by \citeauthor{gwbook}
\cite{gwbook} for an excellent introduction.
%
%
%
%
In this section, we show how to use this theory to model \quickcheck's
distribution of constructors.
\looseness=-1


We start by analyzing the generation process for the \ensuremath{\Conid{Node}} constructors in the
data type \ensuremath{\Conid{Tree}} as described by the generators in Figure \ref{gen:derive} and
\ref{fig:mega}.
From the code, we can observe that the stochastic process they encode satisfies
the following assumptions (which coincide with the assumptions of Galton-Watson
branching processes):
%
\begin{inparaenum}[i)]
\item With a certain probability, it starts with some initial \ensuremath{\Conid{Node}} constructor.
\item At any step, the probability of generating a \ensuremath{\Conid{Node}} is not affected by the
  \ensuremath{\Conid{Node}}s generated before or after.
\item The probability of generating a \ensuremath{\Conid{Node}} is independent of where in the tree
that constructor is about to be placed.
\end{inparaenum}

The original Galton-Wat\-son process is a simple stochastic process that
counts the population sizes at different points in time called
\emph{generations}.
For our purposes, populations consist of \ensuremath{\Conid{Node}} constructors, and generations
are obtained by selecting tree levels.

%
%
%
%
%
%

Figure \ref{fig:gw:simple} illustrates a possible generated value.
It starts by generating a \ensuremath{\Conid{Node}} constructor at generation (i.e., depth) zero
($G_0$), then another two \ensuremath{\Conid{Node}} constructors as left and right subtrees in
generation one ($G_1$), etc.
(Dotted edges denote further constructors which are not drawn, as they are not
essential for the point being made.)
%
%
This process repeats until the population of \ensuremath{\Conid{Node}} constructors becomes extinct
or stable, or alternatively grows forever.

\begin{wrapfigure}{r}{0.25\textwidth}
  \vspace{-8pt}
  \begin{framed}
  \vspace{-3pt}
  \newcommand{\eX}{\dimexpr\fontcharht\font`X\relax}
  \newcommand{\tux}{\includegraphics[height=1.5\eX]{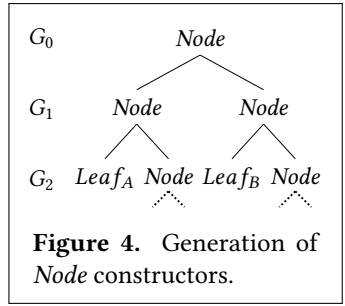}}
  \tikzstyle{densely dotted}=[dash pattern=on \pgflinewidth off 1pt]
  \centering
  \resizebox {4cm} {!} {
    \hspace{-10pt}
  \begin{tikzpicture}
      [level 1+/.style={level distance=1cm, sibling distance=-0.1cm}]
    \Tree
    [.\node(Z0){\ensuremath{\Conid{Node}}};
      [.\node(Z1l){\ensuremath{\Conid{Node}}};
        \node(Z2l1){\ensuremath{Leaf_A}};
        \node(Z2l2){\ensuremath{\Conid{Node}}};
      ]
      [.\node(Z1r){\ensuremath{\Conid{Node}}};
        \node(Z2l3){\ensuremath{Leaf_B}};
        \node(Z2l4){\ensuremath{\Conid{Node}}};
      ]
    ]
    \node[below = -0.15 of Z2l2] (More2) {} ;
    \node[below = -0.15 of Z2l4] (More4) {} ;
    \node[below left  = 0.2 and -0.25 of Z2l2] (More21) {} ;
    \node[below right = 0.2 and -0.25 of Z2l2] (More22) {} ;
    \node[below left  = 0.2 and -0.25 of Z2l4] (More41) {} ;
    \node[below right = 0.2 and -0.25 of Z2l4] (More42) {} ;
    \draw[densely dotted, thick](More2.center)--(More21);
    \draw[densely dotted, thick](More2.center)--(More22);
    \draw[densely dotted, thick](More4.center)--(More41);
    \draw[densely dotted, thick](More4.center)--(More42);
    \node[left = 1.5cm of Z0]         (Z0Text) {$G_0$};
    \node[below = 0.5cm of Z0Text]  (Z1Text) {$G_1$};
    \node[below = 0.5cm of Z1Text]  (Z2Text) {$G_2$};
  \end{tikzpicture}
  }
\vspace{-22pt}
\caption{\label{fig:gw:simple} Generation of \ensuremath{\Conid{Node}} constructors.}
\vspace{-5pt}
\end{framed}
\vspace{-12pt}
\end{wrapfigure}

The mathematics behind the Galton-Watson process allows us to predict the expected
number of offspring at the \emph{n}th-generation, i.e., the number of \ensuremath{\Conid{Node}}
constructors at depth \emph{n} in the generated tree.
Formally, we start by introducing the random variable $R$ to denote the number
of \ensuremath{\Conid{Node}} constructors in the next generation generated by a \ensuremath{\Conid{Node}} constructor
in this generation---the $R$ comes from ``reproduction'' and the reader can
think it as a \ensuremath{\Conid{Node}} constructor reproducing \ensuremath{\Conid{Node}} constructors.
To be a bit more general, let us consider the \ensuremath{\Conid{Tree}} random generator
automatically generated using \derive (Figure \ref{gen:derive}), but where the
probability of choosing between any constructor is no longer uniform.
Instead, we have a $p_C$ probability of choosing the constructor \ensuremath{\Conid{C}}.
%
%
These probabilities are external parameters of the prediction mechanism, and
Section \ref{sec:implementation} explains how they can later be instantiated
with actual values found by optimization, enabling the user to tune the
generated distribution.

%
%
%
We note $p_{Leaf}$ as the probability of generating a leaf of any kind, i.e.,
$p_{Leaf} = p_{LeafA} + p_{LeafB} + p_{LeafC}$.
In this setting, and assuming a parent constructor \ensuremath{\Conid{Node}}, the probabilities of
generating $R$ numbers of \ensuremath{\Conid{Node}} offspring in the next generation (i.e., in the
recursive calls of \ensuremath{\Varid{arbitrary}}) are as follows:
%
%
%
\begin{align*}
  &P(R = 0) = p_{Leaf} \cdot p_{Leaf}\\
  &P(R = 1) = p_{Node} \cdot p_{Leaf} + p_{Leaf} \cdot p_{Node}
    = 2 \cdot p_{Node} \cdot p_{Leaf} \\
  &P(R = 2) = p_{Node} \cdot p_{Node}
\end{align*}
%
One manner to understand the equations above is by considering what \quickcheck
does when generating the subtrees of a given node.
For instance, the cases when generating exactly one \ensuremath{\Conid{Node}} as descendant ($P(R =
1)$) occurs in two situations: when the left subtree is a \ensuremath{\Conid{Node}} and the right
one is a \ensuremath{\Conid{Leaf}}; and viceversa.
The probability for those events to occur is $p_{Node} * p_{Leaf}$ and $p_{Leaf}
* p_{Node}$, respectively.
Then, the probablity of having exactly one \ensuremath{\Conid{Node}} as a descendant is given by
the sum of the probability of both events---the other cases follow a similar
reasoning.

Now that we have determined the distribution of $R$, we proceed to introduce the
random variables $G_n$ to denote the population of \ensuremath{\Conid{Node}} constructors in the
$n$th generation.
We write $\xi^n_i$ for the random variable which captures the number of
(offspring) \ensuremath{\Conid{Node}} constructors at the \emph{n}th generation produced by the
$i$th \ensuremath{\Conid{Node}} constructor at the \emph{(n-1)}th generation.
It is easy to see that it must be the case that
$
G_n = \xi^n_1 + \xi^n_2 + \cdots + \xi^n_{G_{n-1}}
$.
%
%
To deduce $E[G_n]$, i.e. the expected number of \ensuremath{\Conid{Node}}s in the $n$th generation,
we apply the (standard) Law of Total Expectation $E[X] = E[E[X | Y]]$
\footnote{
$E[X | Y]$ is a function on the random variable $Y$, i.e., $E[X | Y] y =
E[ X | Y = y]$ and therefore it is a random variable itself.
In this light, the law says that if we observe the expectations of $X$ given the
different $y_s$, and then we do the expectation of all those values, then we have
the expectation of $X$.  }
with $X =
G_n$ and $Y = G_{n-1}$ to obtain:
%
%
\begin{align}
E[G_n] = E [ E [G_n | G_{n-1}] ]. \label{zn:ex}
\end{align}
%
%
By expanding $G_n$, we deduce that:
%
\begin{align*}
  E[G_n | G_{n-1}]
  &= E[\xi^n_1 + \xi^n_2 +\! \cdots\! + \xi^n_{G_{n-1}} | G_{n-1}]\\
  &= E[\xi^n_1 | G_{n-1}] + E[\xi^n_2 | G_{n-1}]
  +\! \cdots\! + E[\xi^n_{G_{n-1}} | G_{n-1}]
  \nonumber
\end{align*}
%
%
%
Since $\xi^n_1$, $\xi^n_2$, ..., and $\xi^n_{G_{n-1}}$ are \emph{all
  governed by the distribution captured by the random variable $R$} (recall the
assumptions at the beginning of the section), we have that:
\begin{equation}
  E[G_n | G_{n-1}]
  = E[R | G_{n-1}] + E[R | G_{n-1}] + \cdots + E[R | G_{n-1}]
  \nonumber
\end{equation}
Since $R$ is \emph{independent of the generation where \ensuremath{\Conid{Node}} constructors
  decide to generate other \ensuremath{\Conid{Node}} constructors}, we have that
\begin{equation}
  E[G_n | G_{n-1}]
  = \underbrace{E[R] + E[R] + \cdots + E[R]}_{G_{n-1} \ \text{times}}
  = E[R]\! \cdot\! G_{n-1}
  \label{zn:m}
\end{equation}
From now on, we introduce $m$ to denote the mean of $R$, i.e., \emph{the mean of
  reproduction}.
Then, by rewriting $m = E[R]$,
%
%
we obtain:
\begin{align}
  E[G_n]
  \stackbin[]{(\ref{zn:ex})}{=}
  E [E[G_n | G_{n-1}]]
  \stackbin[]{(\ref{zn:m})}{=}
  E [m\! \cdot\! G_{n-1}]
  \stackbin[]{\text{m is constant}}{=}
  E[G_{n-1}]\! \cdot\! m
  \nonumber
\end{align}
By unfolding this recursive equation many times, we obtain:
%
\begin{equation}
  E[G_n] = E[G_0]\! \cdot\! m^n \label{zn:final}
\end{equation}

%
\noindent
As the equation indicates, the expected number of \ensuremath{\Conid{Node}} constructors at the
\emph{n}th generation is affected by the mean of reproduction.
Although we obtained this intuitive result using a formalism that may look
overly complex, it is useful to understand the methodology used here.
In the next section, we will derive the main result of this work following the
same reasoning line under a more general scenario.
%

We can now also predict the total expected number of individuals {\em up to} the
\emph{n}th generation.
For that purpose, we introduce the random variable $P_n$ to denote the
population of \ensuremath{\Conid{Node}} constructors up to the \emph{n}th generation.
It is then easy to see that $P_n = \sum_{i=0}^n G_i$ and consequently:
\begin{align}
  E[P_n]
  \!=\! \sum_{i=0}^n E[G_i]
  \stackbin[]{(\ref{zn:final})}{=}
  \sum_{i=0}^n E[G_0] \cdot m^i
  \stackbin[]{}{=}
  E[G_0]\! \cdot\! \left( \frac{1\! -\! m^{n+1} }{1\! -\! m} \right)
\label{simple:geo}
\end{align}
where the last equality holds by the geometric series definition.
This is the general formula provided by the Galton-Watson process.
In this case, the mean of reproduction for \ensuremath{\Conid{Node}} is given by:
%
\begin{equation}
  m = E[R] = \sum_{k = 0}^{2} k \cdot P(R = k) = 2 \cdot p_{Node}
  \label{exp:R} 
\end{equation}
%
%
%
By (\ref{simple:geo}) and (\ref{exp:R}), the expected number of \ensuremath{\Conid{Node}}
constructors up to generation $n$ is given by the following formula:
\begin{align}
  E[P_n]\! =\! E[G_0]\! \cdot\!
  \left( \frac{1\! -\! m^{n+1}}
  {1\! -\! m} \right)
  \!=\! p_{\ensuremath{\Conid{Node}}}\! \cdot\!
  \left( \frac{1\! -\! (2 \cdot p_{\ensuremath{\Conid{Node}}})^{n+1}}
  {1\! -\! 2\! \cdot\! p_{\ensuremath{\Conid{Node}}}} \right)
  \nonumber
\end{align}
%
%
%
If we apply the previous formula to predict the distribution of constructors
induced by \megadeth in Figure \ref{fig:mega}, where $p_{LeafA} = p_{LeafB} = p_{LeafC} = p_{Node} =
0.25$, we obtain an expected number of \ensuremath{\Conid{Node}} constructors up to level 10 of
0.4997, which denotes a distribution highly biased towards small values, since
we can only produce further subterms by producing \ensuremath{\Conid{Node}}s.
However, if we set $p_{LeafA} = p_{LeafB} = p_{LeafC} = 0.1$ and $p_{Node} =
0.7$, we can predict that, as expected, our general random generator will
generate much bigger trees, containing an average number of 69.1173 \ensuremath{\Conid{Node}}s up
to level 10!
Unfortunately, we cannot apply this reasoning to predict the distribution of
constructors for derived generators for ADTs with more than one non-terminal
constructor.
For instance, let us consider the following data type definition:
%
%
\begin{hscode}\SaveRestoreHook
\column{B}{@{}>{\hspre}l<{\hspost}@{}}%
\column{3}{@{}>{\hspre}l<{\hspost}@{}}%
\column{E}{@{}>{\hspre}l<{\hspost}@{}}%
\>[3]{}\mathbf{data}\;\Conid{Tree'}\mathrel{=}\Conid{Leaf}\mid Node_A\;\Conid{Tree'}\;\Conid{Tree'}\mid Node_B\;\Conid{Tree'}{}\<[E]%
\ColumnHook
\end{hscode}\resethooks
In this case, we need to separately consider that a \ensuremath{Node_A} can generate not
only \ensuremath{Node_A} but also \ensuremath{Node_B} offspring (similarly with \ensuremath{Node_B}).
A stronger mathematical formalism is needed.
%
%
The next section explains how to predict the generation of this kind of data
types by using an extension of Galton-Waston processes known as \emph{multi-type
  branching processes}.



\section{Multi-Type Branching Processes}
\label{sec:bp2}


\begin{figure}[b]
\vspace{-10pt}
\begin{framed}
\vspace{-10pt}
\begin{hscode}\SaveRestoreHook
\column{B}{@{}>{\hspre}l<{\hspost}@{}}%
\column{3}{@{}>{\hspre}l<{\hspost}@{}}%
\column{5}{@{}>{\hspre}l<{\hspost}@{}}%
\column{7}{@{}>{\hspre}l<{\hspost}@{}}%
\column{9}{@{}>{\hspre}c<{\hspost}@{}}%
\column{9E}{@{}l@{}}%
\column{12}{@{}>{\hspre}l<{\hspost}@{}}%
\column{14}{@{}>{\hspre}l<{\hspost}@{}}%
\column{19}{@{}>{\hspre}c<{\hspost}@{}}%
\column{19E}{@{}l@{}}%
\column{22}{@{}>{\hspre}l<{\hspost}@{}}%
\column{E}{@{}>{\hspre}l<{\hspost}@{}}%
\>[3]{}\mathbf{instance}\;\Conid{Arbitrary}\;\Conid{Tree'}\;\mathbf{where}{}\<[E]%
\\
\>[3]{}\hsindent{2}{}\<[5]%
\>[5]{}\Varid{arbitrary}\mathrel{=}\Varid{sized}\;\Varid{gen}\;\mathbf{where}{}\<[E]%
\\
\>[5]{}\hsindent{2}{}\<[7]%
\>[7]{}\Varid{gen}\;\mathrm{0}\mathrel{=}\Varid{pure}\;\Conid{Leaf}{}\<[E]%
\\
\>[5]{}\hsindent{2}{}\<[7]%
\>[7]{}\Varid{gen}\;\Varid{n}{}\<[14]%
\>[14]{}\mathrel{=}\Varid{chooseWith}{}\<[E]%
\\
\>[7]{}\hsindent{2}{}\<[9]%
\>[9]{}[\mskip1.5mu {}\<[9E]%
\>[12]{}(p_{Leaf}{}\<[19]%
\>[19]{},{}\<[19E]%
\>[22]{}\Varid{pure}\;\Conid{Leaf}){}\<[E]%
\\
\>[7]{}\hsindent{2}{}\<[9]%
\>[9]{},{}\<[9E]%
\>[12]{}(p_{NodeA}{}\<[19]%
\>[19]{},{}\<[19E]%
\>[22]{}Node_A\mathop{\langle \texttt{\$} \rangle}\Varid{gen}\;(n{-}1)\mathop{\langle \ast \rangle}\Varid{gen}\;(n{-}1)){}\<[E]%
\\
\>[7]{}\hsindent{2}{}\<[9]%
\>[9]{},{}\<[9E]%
\>[12]{}(p_{NodeB}{}\<[19]%
\>[19]{},{}\<[19E]%
\>[22]{}Node_B\mathop{\langle \texttt{\$} \rangle}\Varid{gen}\;(n{-}1))\mskip1.5mu]{}\<[E]%
\ColumnHook
\end{hscode}\resethooks
\vspace{-15pt}
\caption{\label{fig:treep:dragen} \dragen generator for \ensuremath{\Conid{Tree'}}}
\vspace{-5pt}
\end{framed}
\end{figure}

In this section, we present the basis for our main contribution: \emph{the
  application of multi-type branching processes to predict the distribution of
  constructors}.
%
We will illustrate the technique by considering the \ensuremath{\Conid{Tree'}} ADT that we
concluded with in the previous section.
%

Before we dive into technicalities, Figure \ref{fig:treep:dragen} shows the
automatically derived generator for \ensuremath{\Conid{Tree'}} that our tool produces.
Our generators depend on the (possibly) different probabilities that
constructors have to be generated---variables \ensuremath{p_{Leaf}}, \ensuremath{p_{NodeA}}, and \ensuremath{p_{NodeB}}.
These probabilities are used by the function \ensuremath{\Varid{chooseWith}\mathbin{::}[\mskip1.5mu (\Conid{Double},\Conid{Gen}\;\Varid{a})\mskip1.5mu]\to \Conid{Gen}\;\Varid{a}}, which picks a random generator of type \ensuremath{\Varid{a}} with an explicitly given
probability from a list.
%
%
This function can be easily expressed by using \quickcheck's primitive
operations and therefore we omit its implementation.
Additionally note that, like \megadeth, our generators use \ensuremath{\Varid{sized}} to limit the
number of recursive calls to ensure termination.
%
%
%
%
%
%
We note that the theory behind branching processes is able to predict the
termination behavior of our generators
%
%
and we could have used this ability to ensure their termination without the need
of a depth limiting mechanism like \ensuremath{\Varid{sized}}.
However, using \ensuremath{\Varid{sized}} provides more control over the obtained generator
distributions.

To predict the distribution of constructors provided by \linebreak \dragen
generators, we introduce a generalization of the previous Galton-Watson
branching process called multi-type Galton-Watson branching process.
This generalization allows us to consider several \emph{kinds of individuals},
i.e., constructors in our setting, to procreate (generate) different \emph{kinds
  of offspring} (constructors).
Additionally, this approach allows us to consider not just one constructor, as
we did in the previous section, but rather to consider all of them at the same
time.

Before we present the mathematical foundations, which follow a similar line of
reasoning as that in Section \ref{sec:bp}, Figure \ref{fig:gw:multi} illustrates
a possible generated value of type \ensuremath{\Conid{Tree'}}.

\begin{wrapfigure}{r}{0.25\textwidth}
\vspace{-5pt}
\begin{framed}
\tikzstyle{densely dotted}=[dash pattern=on \pgflinewidth off 1pt]
\centering
\hspace{5pt}
  \begin{tikzpicture}[level 1+/.style={level distance=0.8cm}]
    \Tree
      [.\node(Z0){\ensuremath{Node_A}};
      [.\node(Z1l){\ensuremath{Node_B}};
        \node(Z2l1){\ensuremath{Node_A}};
      ]
      [.\node(Z1r){\ensuremath{Node_A}};
        \node(Z2l3){\ensuremath{Node_B}};
        \node(Z2l4){\ensuremath{\Conid{Leaf}}};
      ]
    ]
    \node[below = -0.15 of Z2l1] (More1) {} ;
    \node[below = -0.15 of Z2l3] (More3) {} ;
    \node[below left  = 0.2 and -0.25 of Z2l1] (More11) {} ;
    \node[below right = 0.2 and -0.25 of Z2l1] (More12) {} ;
    \node[below = 0.25                of Z2l3] (More31) {} ;
    \draw[densely dotted, thick](More1.center)--(More11);
    \draw[densely dotted, thick](More1.center)--(More12);
    \draw[densely dotted, thick](More3.center)--(More31);
  \end{tikzpicture}
\vspace{-10pt}
\caption{\label{fig:gw:multi} A generated value of type \ensuremath{\Conid{Tree'}}.}
\end{framed}
\vspace{-10pt}
\end{wrapfigure}


In the generation process, it is assumed that \emph{the {kind} (i.e., the
  constructor) of the parent might affect the probabilities of reproducing
  (generating) offspring of a certain {kind}}.
Observe that this is the case for a wide range of derived ADT generators, e.g.,
choosing a terminal constructor (e.g., \ensuremath{\Conid{Leaf}}) affects the probabilities of
generating non-terminal ones (by setting them to zero).
The population at the \emph{n}th generation is then characterized as a vector of
random variables $G_n = (G_n^1, G_n^2, \cdots, G_n^d)$, where $d$ is the number
of different kinds of constructors.
%
Each random variable $G_n^i$ captures the number of occurrences of the
\emph{i}th-constructor of the ADT at the \emph{n}th generation.
Essentially, $G_n$ ``groups'' the population at level \emph{n} by the
constructors of the ADT.
%
%
By estimating the expected shape of the vector $G_n$, it is possible to obtain
the expected number of constructors at the \emph{n}th generation.
%
%
Specifically, we have that
$E[G_n] = (E[G_n^1], E[G_n^2], \cdots, E[G_n^d])$.
To deduce $E[G_n]$, we focus on deducing each component of the vector.
%
%
%
%
%
%
%
%

As explained above, the reproduction behavior is determined by the kind of the
individual.
In this light, we introduce random variable $R_{ij}$ to denote a parent $i$th
constructor reproducing a $j$th constructor.
As we did before, we apply the equation $E[X] = E[E[X | Y]]$ with $X = G_n^j$
and $Y = G_{n-1}$ to obtain $ E[G_n^j] = E [ E[G_n^j | G_{n-1}]]$.
%
%
To calculate the expected number of $j$th constructors at the level $n$ produced
by the constructors present at level $(n-1)$, i.e., $E[G_n^j | G_{(n-1)}]$, it
is enough to count the expected number of children of kind $j$ produced by the
different parents of kind $i$, i.e., $E[R_{ij}]$, times the amount of parents of
kind $i$ found in the level $(n-1)$, i.e., $G_{(n-1)}^i$.
This result is expressed by the following equation marked as ($\star$), and is
formally verified%
\ifbool{EXTENDED}{
  in the Appendix \ref{app:expectation}.%
}{
  in the supplementary material.%
}

\begin{equation}
  E[G_n^j | G_{n-1}]
  \stackbin{(\star)}{=}
  \sum_{i=1}^{d} G_{(n-1)}^i \! \cdot\! E[R_{ij}]
  = \sum_{i=1}^{d} G_{(n-1)}^i  \! \cdot\! m_{ij}
  \label{comp:vector}
\end{equation}

Similarly as before, we rewrite $E[R_{ij}]$ as $m_{ij}$, which now represents a
single expectation of reproduction indexed by the kind of both the parent and
child constructor.

%
%
\paragraph{Mean matrix of constructors}
In the previous section, $m$ was the expectation of reproduction of a single
constructor.
Now we have $m_{ij}$ as the expectation of reproduction indexed by the parent
and child constructor.
In this light, we define $M_C$, the \emph{mean matrix of constructors} (or mean
matrix for simplicity) such that each $m_{ij}$ stores the expected number of
$j$th constructors generated by the $i$th constructor.
$M_C$ is a parameter of the Galton-Watson multi-type process and can be built at
compile-time using statically known type information.
We are now able to deduce $E[G_n^j]$.
\begin{align*}
  E[G_n^j]
  &= E[E[G_n^j | G_{n-1}]]
    \stackbin[]{(\ref{comp:vector})}{=}
  E \left[ \sum_{i=1}^{d} G_{(n-1)}^i \! \cdot\! m_{ij} \right] \\
  &= \sum_{i=1}^{d} E[G_{(n-1)}^i \! \cdot\! m_{ij}]
  = \sum_{i=1}^{d} E[G_{(n-1)}^i] \! \cdot\! m_{ij}
\end{align*}
Using this last equation, we can rewrite $E[G_n]$ as follows.
\begin{align*}
  E[G_n]
  = \left( \sum_{i=1}^{d} E[ G_{(n-1)}^1]\! \cdot\! m_{i1},
  \cdots,
  \sum_{i=1}^{d} E[ G_{(n-1)}^d]\! \cdot\! m_{id} \right)
\end{align*}
By linear algebra, we can rewrite the vector above as the matrix multiplication
${E[G_n]^T = E[G_{n-1}]^T\!\cdot\! M_C}$.
By repeatedly unfolding this definition, we obtain that:
\begin{align}
  E[G_n]^T = E[G_0]^T \cdot (M_C)^n
  \label{eq:matrix}
\end{align}
This equation is a generalization of (\ref{zn:final}) when considering many
constructors.
As we did before, we introduce a random variable $P_n = \sum_{i=0}^n G_i$ to
denote the population up to the \emph{n}th generation.
It is now possible to obtain the expected population of all the constructors but
in a clustered manner:
\begin{align}
  E[P_n]^T
  \stackbin[]{}{=} E \left[ \sum_{i=0}^n G_i \right]^T
  \stackbin[]{}{=} \sum_{i=0}^n E[G_i]^T
  \stackbin[]{(\ref{eq:matrix})}{=} \sum_{i=0}^n E[G_0]^T\! \cdot\! (M_C)^n
  \label{eq:sum:multi}
\end{align}
It is possible to write the resulting sum as the closed formula:
\begin{align}
  \label{eq:total:multi}
  E[P_n]^T
  =
  E[G_0]^T \cdot \left( \frac{I - (M_C)^{n+1}}{I - M_C} \right)
\end{align}
where $I$ represents the identity matrix of the appropriate size.
Note that equation $(\ref{eq:total:multi})$ only holds when $(I - M_C)$ is
non-singular, however, this is the usual case.
%
%
%
%
When $(I - M_C)$ is singular, we resort to using equation $(\ref{eq:sum:multi})$
instead.
Without losing generality, and for simplicity, we consider equations
$(\ref{eq:sum:multi})$ and $(\ref{eq:total:multi})$ as interchangeable.
They are the general formulas for the Galton-Watson multi-type branching
processes.

Then, to predict the distribution of our \ensuremath{\Conid{Tree'}} data type example, we proceed
to build its mean matrix $M_C$.
For instance, the mean number of \ensuremath{\Conid{Leaf}}s generated by a \ensuremath{Node_A} is:
\begin{align}
  m_{\ensuremath{Node_A},\ensuremath{\Conid{Leaf}}}
  \quad&=\quad \underbrace{1 \cdot p_{\ensuremath{\Conid{Leaf}}} \cdot
         p_{\ensuremath{Node_A}} + 1 \cdot p_{\ensuremath{\Conid{Leaf}}} \cdot
         p_{\ensuremath{Node_B}}}_\text{$One \ \mathrm{\ensuremath{\Conid{Leaf}}}$ as
         left-subtree}
         \nonumber \\
       &+\quad \underbrace{1 \cdot p_{\ensuremath{Node_A}} \cdot p_{\ensuremath{\Conid{Leaf}}} + 1
         \cdot p_{\ensuremath{Node_B}} \cdot p_{\ensuremath{\Conid{Leaf}}}}_\text{$One \
         \mathrm{\ensuremath{\Conid{Leaf}}}$ as right-subtree}
         \nonumber \\
       &+\quad \underbrace{2 \cdot p_{\ensuremath{\Conid{Leaf}}} \cdot
         p_{\ensuremath{\Conid{Leaf}}}}_\text{$\mathrm{\ensuremath{\Conid{Leaf}}}$ as left- and
         right-subtree}
         \nonumber \\
       &=\quad \ 2 \cdot p_{\ensuremath{\Conid{Leaf}}}
         \label{mNodeLeaf}
\end{align}
The rest of $M_C$ can be similarly computed, obtaining:
%
\begin{equation}
  M_C
  = \qquad\quad
  \mkern-5mu
  \begin{tikzpicture}[baseline=-0.65ex]
    \matrix[
    matrix of math nodes,
    column sep=0ex,
    left delimiter={[},
    right delimiter={]}
    ] (m)
    {
      0               & 0                & 0\\
      2 \cdot p_{\ensuremath{\Conid{Leaf}}} & 2 \cdot p_{\ensuremath{Node_A}} &2 \cdot p_{\ensuremath{Node_B}}\\
       p_{\ensuremath{\Conid{Leaf}}}  &  p_{\ensuremath{Node_A}} &  p_{\ensuremath{Node_B}}\\
    };
    \node[above,text depth=1pt] at (m-1-1.north) {$\scriptstyle \ensuremath{\Conid{Leaf}}$};
    \node[above,text depth=1pt] at (m-1-2.north) {$\scriptstyle \ensuremath{Node_A}$};
    \node[above,text depth=1pt] at (m-1-3.north) {$\scriptstyle \ensuremath{Node_B}$};
    \node[left,overlay] at ([xshift=-4.5ex]m-1-1.west) {$\scriptstyle \ensuremath{\Conid{Leaf}}$};
    \node[left,overlay] at ([xshift=-1.5ex]m-2-1.west) {$\scriptstyle \ensuremath{Node_A}$};
    \node[left,overlay] at ([xshift=-3.0ex]m-3-1.west) {$\scriptstyle \ensuremath{Node_B}$};
  \end{tikzpicture}\mkern-5mu
  \label{eq:treepmatrix}
\end{equation}
%
%
%
Note that the first row, corresponding to the \ensuremath{\Conid{Leaf}} constructor, is filled with
zeros.
This is because \ensuremath{\Conid{Leaf}} is a terminal constructor, i.e., it cannot generate
further subterms of any kind.%
\footnote{The careful reader may notice that there is a pattern in the mean
  matrix if inspected together with the definition of \ensuremath{\Conid{Tree'}}.
  We prove in Section \ref{sec:opt} that each $m_{ij}$ can be automatically
  calculated by simply exploiting type information.}
%
%
%
%

With the mean matrix in place, we define $E[G_0]$ (the initial vector of mean
probabilities) as $(p_{\ensuremath{\Conid{Leaf}}}, p_{\ensuremath{Node_A}}, p_{\ensuremath{Node_B}})$.
By applying (\ref{eq:total:multi}) with $E[G_0]$ and $M_C$, we can predict the
expected number of generated \emph{non-terminal} \ensuremath{Node_A} constructors (and
analogously \ensuremath{Node_B}) with a size parameter $n$ as follows:
%
\begin{align}
  E[\mathrm{\ensuremath{Node_A}}]
  \!&=\! \left(E[P_{n-1}]^T \right)\!.\mathrm{\ensuremath{Node_A}}
    \!=\! \left(E[G_0]^T\! \cdot\!
      \left( \frac{I\! -\! (M_C)^{n}}{I\! -\! M_C} \!\right)\! \right)\!.\mathrm{\ensuremath{Node_A}}
    \nonumber
\end{align}
Function $(\_).C$ simply projects the value corresponding to constructor $C$
from the population vector.
It is very important to note that the sum only includes the population up to
level ($n-1$).
This choice comes from the fact that our \quickcheck generator can choose
between only terminal constructors at the last generation level (recall that
\ensuremath{\Varid{gen}\;\mathrm{0}} generates only \ensuremath{\Conid{Leaf}}s in Figure \ref{fig:treep:dragen}).
%
%
As an example, if we assign our generation probabilities for \ensuremath{\Conid{Tree'}} as
$p_{Leaf} \mapsto 0.2$, $p_{\ensuremath{Node_A}}\mapsto 0.5$ and $p_{\ensuremath{Node_B}} \mapsto 0.3$,
then the formula predicts that our \quickcheck generator with a size parameter
of 10 will generate on average 21.322 \ensuremath{Node_A}s and 12.813 \ensuremath{Node_B}s.
This result can easily be verified by sampling a large number of values with a
generation size of 10, and then averaging the number of generated \ensuremath{Node_A}s and
\ensuremath{Node_B}s across the generated values.

In this section, we obtain a prediction of the expected number of non-terminal
constructors generated by \dragen generators.
To predict terminal constructors, however, requires a special treatment as
discussed in the next section.


%
%


\section{Terminal constructors}
\label{sec:terminals}

In this section we introduce the special treatment required to predict the
generated distribution of terminal constructors, i.e. constructors with no
recursive arguments.

Consider the generator in Figure \ref{fig:treep:dragen}.
It generates terminal constructors in two situations, i.e., in the definition of
\ensuremath{\Varid{gen}\;\mathrm{0}} and \ensuremath{\Varid{gen}\;\Varid{n}}.
In other words, the random process introduced by our generators can be
considered to be composed of two independent parts when it comes to terminal
constructors---refer to%
\ifbool{EXTENDED}{
  Appendix \ref{app:terminals}
}{
  the supplementary material
}%
for a graphical interpretation.
In principle, the number of terminal constructors generated by the stochastic
process described in \ensuremath{\Varid{gen}\;\Varid{n}} is captured by the multi-type branching process
formulas.
However, to predict the expected number of terminal constructors generated by
exercising \ensuremath{\Varid{gen}\;\mathrm{0}}, we need to separately consider a random process that
\emph{only generates terminal constructors} in order to terminate.
For this purpose, and assuming a maximum generation depth $n$, we need to
calculate the number of terminal constructors required to stop the generation
process at the recursive arguments of each non-terminal constructor at level
($n-1$).
In our \ensuremath{\Conid{Tree'}} example, this corresponds to two \ensuremath{\Conid{Leaf}}s for every \ensuremath{Node_A} and
one \ensuremath{\Conid{Leaf}} for every \ensuremath{Node_B} constructor at level ($n-1$).
\looseness=-1

Since both random processes are independent, to predict the overall expected
number of terminal constructors, we can simply add the expected number of
terminal constructors generated in each one of them.
%
%
%
%
%
%
%
%
%
%
%
%
Recalling our previous example, we obtain the following formula for \ensuremath{\Conid{Tree'}}
terminals as follows:
\begin{align*}
  E[\ensuremath{\Conid{Leaf}}]
  = \underbrace{\left( E[P_{n-1}]^T  \right)\!.
    \mathrm{\ensuremath{\Conid{Leaf}}}}_{\text{branching process}}
  &+ \underbrace{2\! \cdot\! \left( E[G_{n-1}]^T \right)\!.
    \mathrm{\ensuremath{Node_A}}}_\text{case (\mbox{\ensuremath{Node_A\;\Conid{Leaf}\;\Conid{Leaf}}})} \\
  &+ \underbrace{1\! \cdot\! \left( E[G_{n-1}]^T \right)\!.
    \mathrm{\ensuremath{Node_B}}}_\text{case (\mbox{\ensuremath{Node_B\;\Conid{Leaf}}})}
\end{align*}
The formula counts the \ensuremath{\Conid{Leaf}}s generated by the multi-type branching process up
to level ($n-1$) and adds the expected number of \ensuremath{\Conid{Leaf}}s generated at the last
level.
%

%
Although we can now predict the expected number of generated \ensuremath{\Conid{Tree'}}
constructors regardless of whether they are terminal or not, this approach only
works for data types with a single terminal constructor.
\begin{figure}[b] 
\vspace{-10pt}
\begin{framed}
\vspace{-10pt}
\begin{hscode}\SaveRestoreHook
\column{B}{@{}>{\hspre}l<{\hspost}@{}}%
\column{3}{@{}>{\hspre}l<{\hspost}@{}}%
\column{5}{@{}>{\hspre}l<{\hspost}@{}}%
\column{7}{@{}>{\hspre}c<{\hspost}@{}}%
\column{7E}{@{}l@{}}%
\column{10}{@{}>{\hspre}l<{\hspost}@{}}%
\column{18}{@{}>{\hspre}l<{\hspost}@{}}%
\column{31}{@{}>{\hspre}l<{\hspost}@{}}%
\column{32}{@{}>{\hspre}l<{\hspost}@{}}%
\column{E}{@{}>{\hspre}l<{\hspost}@{}}%
\>[B]{}\mathbf{instance}\;\Conid{Arbitrary}\;Tree'\!'\;\mathbf{where}{}\<[E]%
\\
\>[B]{}\hsindent{3}{}\<[3]%
\>[3]{}\Varid{arbitrary}\mathrel{=}\Varid{sized}\;\Varid{gen}\;\mathbf{where}{}\<[E]%
\\
\>[3]{}\hsindent{2}{}\<[5]%
\>[5]{}\Varid{gen}\;\mathrm{0}\mathrel{=}\Varid{chooseWith}{}\<[E]%
\\
\>[5]{}\hsindent{2}{}\<[7]%
\>[7]{}[\mskip1.5mu {}\<[7E]%
\>[10]{}(p_{LeafA}^*,\Varid{pure}\;Leaf_A),{}\<[31]%
\>[31]{}(p_{LeafB}^*,\Varid{pure}\;Leaf_B)\mskip1.5mu]{}\<[E]%
\\
\>[3]{}\hsindent{2}{}\<[5]%
\>[5]{}\Varid{gen}\;\Varid{n}\mathrel{=}\Varid{chooseWith}{}\<[E]%
\\
\>[5]{}\hsindent{2}{}\<[7]%
\>[7]{}[\mskip1.5mu {}\<[7E]%
\>[10]{}(p_{LeafA},{}\<[18]%
\>[18]{}\Varid{pure}\;Leaf_A),{}\<[32]%
\>[32]{}(p_{LeafB},\Varid{pure}\;Leaf_B){}\<[E]%
\\
\>[5]{}\hsindent{2}{}\<[7]%
\>[7]{},{}\<[7E]%
\>[10]{}(p_{NodeA},{}\<[18]%
\>[18]{}Node_A\mathop{\langle \texttt{\$} \rangle}\Varid{gen}\;(n{-}1)\mathop{\langle \ast \rangle}\Varid{gen}\;(n{-}1)){}\<[E]%
\\
\>[5]{}\hsindent{2}{}\<[7]%
\>[7]{},{}\<[7E]%
\>[10]{}(p_{NodeB},{}\<[18]%
\>[18]{}Node_B\mathop{\langle \texttt{\$} \rangle}\Varid{gen}\;(n{-}1))\mskip1.5mu]{}\<[E]%
\ColumnHook
\end{hscode}\resethooks
\vspace{-15pt}
\caption{\label{fig:treepp} Derived generator for \ensuremath{Tree'\!'}}
\vspace{-5pt}
\end{framed}
\end{figure}
%
%
%
%
%
If we have a data type with multiple terminal constructors, we have to consider
the probabilities of choosing each one of them when filling the recursive
arguments of non-terminal constructors at the previous level.
For instance, consider the following ADT:
\begin{hscode}\SaveRestoreHook
\column{B}{@{}>{\hspre}l<{\hspost}@{}}%
\column{3}{@{}>{\hspre}l<{\hspost}@{}}%
\column{E}{@{}>{\hspre}l<{\hspost}@{}}%
\>[3]{}\mathbf{data}\;Tree'\!'\mathrel{=}Leaf_A\;\!\vert\!\;Leaf_B\;\!\vert\!\;Node_A\;Tree'\!'\;Tree'\!'\;\!\vert\!\;Node_B\;Tree'\!'{}\<[E]%
\ColumnHook
\end{hscode}\resethooks
%
%

Figure \ref{fig:treepp} shows the corresponding \dragen generator for \ensuremath{Tree'\!'}.
%
%
Note there are two sets of probabilities to choose terminal nodes, one for each
random process.
The \ensuremath{p_{LeafA}^*} and \ensuremath{p_{LeafB}^*} probabilities are used to choose between terminal
constructors at the last generation level.
These probabilities preserve the same proportion as their non-starred versions,
i.e., they are normalized to form a probability distribution:
%
%
\begin{align*}
  p_{\ensuremath{Leaf_A}}^* = \frac{p_{\ensuremath{Leaf_A}}}{p_{\ensuremath{Leaf_A}} + p_{\ensuremath{Leaf_B}}}
  \hspace{20pt}
  p_{\ensuremath{Leaf_B}}^* = \frac{p_{\ensuremath{Leaf_B}}}{p_{\ensuremath{Leaf_A}} + p_{\ensuremath{Leaf_B}}}
\end{align*}
In this manner, we can use the same generation probabilities for terminal
%
%
constructors in both random processes---therefore reducing the complexity of our
prediction engine implementation (described in Section
\ref{sec:implementation}).
%
%

To compute the overall expected number of terminals, we need to predict the
expected number of terminal constructors at the last generation level which
could be descendants of non-terminal constructors at level $(n-1)$.
More precisely:
%
%
%
\begin{align*}
  E[\ensuremath{Leaf_A}]
  = \underbrace{
    \left( E[P_{n-1}]^T \right)\!.\mathrm{\ensuremath{Leaf_A}}
    }_{\text{branching process}}
  &+ \underbrace{
    2\! \cdot\! p_{\ensuremath{Leaf_A}}^*\! \cdot\! \left( E[G_{n-1}]^T \right)\!.\mathrm{\ensuremath{Node_A}}\
    }_{\text{expected leaves to fill $\ensuremath{Node_A}s$}}\\
  &+ \underbrace{
    1\! \cdot\! p_{\ensuremath{Leaf_A}}^*\! \cdot\! \left( E[G_{n-1}]^T \right)\!.\mathrm{\ensuremath{Node_B}}
    }_{\text{expected leaves to fill $\ensuremath{Node_B}s$}}
\end{align*}
where the case of $E[\ensuremath{Leaf_B}]$ follows analogously.
%
%


\section{Mutually-recursive and composite ADTs}
\label{sec:opt}

%
In this section, we introduce some extensions to our model that allow us to
derive \dragen generators for data types found in existing off-the-shelf Haskell
libraries.
We start by showing how multi-type branching processes
%
%
naturally extend to mutually-recursive ADTs.
Consider the mutually recursive ADTs \ensuremath{T_{1}} and \ensuremath{T_{2}} with their automatically
derived generators shown in Figure \ref{fig:t1t2}.
%
%
%
%
%
%
%
%
%
\begin{figure}[t]
  \begin{framed}
    \vspace{-10pt}
    \begin{hscode}\SaveRestoreHook
\column{B}{@{}>{\hspre}l<{\hspost}@{}}%
\column{5}{@{}>{\hspre}l<{\hspost}@{}}%
\column{7}{@{}>{\hspre}l<{\hspost}@{}}%
\column{9}{@{}>{\hspre}l<{\hspost}@{}}%
\column{11}{@{}>{\hspre}l<{\hspost}@{}}%
\column{14}{@{}>{\hspre}l<{\hspost}@{}}%
\column{E}{@{}>{\hspre}l<{\hspost}@{}}%
\>[5]{}\mathbf{data}\;T_{1}\mathrel{=}\Conid{A}\mid \Conid{B}\;T_{1}\;T_{2}{}\<[E]%
\\
\>[5]{}\mathbf{data}\;T_{2}\mathrel{=}\Conid{C}\mid \Conid{D}\;T_{1}{}\<[E]%
\\[\blanklineskip]%
\>[5]{}\mathbf{instance}\;\Conid{Arbitrary}\;T_{1}\;\mathbf{where}{}\<[E]%
\\
\>[5]{}\hsindent{2}{}\<[7]%
\>[7]{}\Varid{arbitrary}\mathrel{=}\Varid{sized}\;\Varid{gen}\;\mathbf{where}{}\<[E]%
\\
\>[7]{}\hsindent{2}{}\<[9]%
\>[9]{}\Varid{gen}\;\mathrm{0}\mathrel{=}\Varid{pure}\;\Conid{A}{}\<[E]%
\\
\>[7]{}\hsindent{2}{}\<[9]%
\>[9]{}\Varid{gen}\;\Varid{n}\mathrel{=}\Varid{chooseWith}{}\<[E]%
\\
\>[9]{}\hsindent{2}{}\<[11]%
\>[11]{}[\mskip1.5mu {}\<[14]%
\>[14]{}(p_{A},\Varid{pure}\;\Conid{A}){}\<[E]%
\\
\>[9]{}\hsindent{2}{}\<[11]%
\>[11]{},{}\<[14]%
\>[14]{}(p_{B},\Conid{B}\mathop{\langle \texttt{\$} \rangle}\Varid{gen}\;(n{-}1)\mathop{\langle \ast \rangle}\Varid{resize}\;(n{-}1)\;\Varid{arbitrary})\mskip1.5mu]{}\<[E]%
\\[\blanklineskip]%
\>[5]{}\mathbf{instance}\;\Conid{Arbitrary}\;T_{2}\;\mathbf{where}{}\<[E]%
\\
\>[5]{}\hsindent{2}{}\<[7]%
\>[7]{}\Varid{arbitrary}\mathrel{=}\Varid{sized}\;\Varid{gen}\;\mathbf{where}{}\<[E]%
\\
\>[7]{}\hsindent{2}{}\<[9]%
\>[9]{}\Varid{gen}\;\mathrm{0}\mathrel{=}\Varid{pure}\;\Conid{C}{}\<[E]%
\\
\>[7]{}\hsindent{2}{}\<[9]%
\>[9]{}\Varid{gen}\;\Varid{n}\mathrel{=}\Varid{chooseWith}{}\<[E]%
\\
\>[9]{}\hsindent{2}{}\<[11]%
\>[11]{}[\mskip1.5mu (p_{C},\Varid{pure}\;\Conid{C}),(p_{D},\Conid{D}\mathop{\langle \texttt{\$} \rangle}\Varid{resize}\;(n{-}1)\;\Varid{arbitrary})\mskip1.5mu]{}\<[E]%
\ColumnHook
\end{hscode}\resethooks
  \vspace{-15pt}
  \caption{\label{fig:t1t2} Mutually recursive types \ensuremath{T_{1}} and \ensuremath{T_{2}} and their
    \dragen generators.}
  \vspace{-5pt}
\end{framed}
\vspace{-10pt}
\end{figure}
%
%
Note the use of the \quickcheck's function \ensuremath{\Varid{resize}\mathbin{::}\Conid{Int}\to \Conid{Gen}\;\Varid{a}\to \Conid{Gen}\;\Varid{a}},
which resets the generation size of a given generator to a new value.
We use it to decrement the generation size at the recursive calls of \ensuremath{\Varid{arbitrary}}
that generate subterms of a mutually recursive data type.
%
%

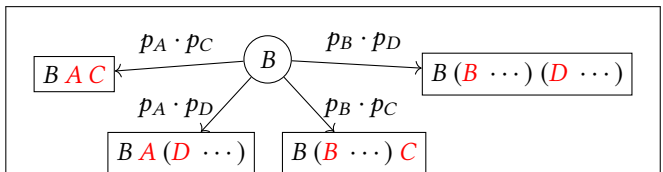
\begin{figure}[b]
  \vspace{-5pt}
  \begin{framed}
  \vspace{-5pt}
  \begin{center}
    \begin{tikzpicture}
      [ level/.style={sibling distance=23mm/#1
      , level distance=12.5mm} ]
      \node [circle,draw] (B){\ensuremath{\Conid{B}}}
      child {node [rectangle,draw, xshift=25pt, yshift=30pt] (ac) {\ensuremath{\Conid{B}\;\textcolor{red}{A}\;\textcolor{red}{C}}}
        edge from parent[draw=none] node[left, yshift=10pt, xshift=25pt]
        {$p_A \cdot p_C\ \ $}}
      child {node [rectangle,draw] (ad) {\ensuremath{\Conid{B}\;\textcolor{red}{A}\;(\textcolor{red}{D}\;\cdots)}}
        edge from parent[draw=none] node[left, yshift=0pt, xshift=0pt]
        {$\ p_A \cdot p_D$}}
      child {node [rectangle,draw] (bc) {\ensuremath{\Conid{B}\;(\textcolor{red}{B}\;\cdots)\;\textcolor{red}{C}}}
        edge from parent[draw=none] node[right, yshift=0pt, xshift=-3pt]
        {$\ \ p_B \cdot p_C$}}
      child {node [rectangle,draw, yshift=30pt] (bd) {\ensuremath{\Conid{B}\;(\textcolor{red}{B}\;\cdots)\;(\textcolor{red}{D}\;\cdots)}}
        edge from parent[draw=none] node[right, yshift=10pt, xshift=-40pt]
        {$\ \ \ \ p_B \cdot p_D$}};
      \draw[transform canvas={xshift=0em}, ->] (B) -- (ac);
      \draw[transform canvas={xshift=0em}, ->] (B) -- (ad);
      \draw[transform canvas={xshift=0em}, ->] (B) -- (bc);
      \draw[transform canvas={xshift=0em}, ->] (B) -- (bd);
    \end{tikzpicture}
  \end{center}
  \vspace{-5pt}
  \caption{\label{fig:BandD} Possible offspring of constructor \ensuremath{\Conid{B}}.} 
  \vspace{-5pt}
  \end{framed}
\end{figure}
The key observation is that we can ignore that \ensuremath{\Conid{A}}, \ensuremath{\Conid{B}}, \ensuremath{\Conid{C}} and \ensuremath{\Conid{D}} are
\emph{constructors belonging to different data types} and just consider each of
them as a kind of offspring on its own.
Figure \ref{fig:BandD} visualizes the possible offspring generated by the
non-terminal constructor \ensuremath{\Conid{B}} (belonging to \ensuremath{T_{1}}) with the corresponding
probabilities as labeled edges.
Following the figure, we obtain the expected number of \ensuremath{\Conid{D}}s generated by \ensuremath{\Conid{B}}
constructors as follows:
\begin{align*}
  m_{BD}
  = 1 \cdot p_A \cdot p_D + 1 \cdot p_B \cdot p_D
  = p_D \cdot (p_A + p_B) = p_D
\end{align*}
Doing similar calculations, we obtain the mean matrix $M_C$ for \ensuremath{\Conid{A}}, \ensuremath{\Conid{B}}, \ensuremath{\Conid{C}},
and \ensuremath{\Conid{D}} as follows:
%
\begin{equation}
  M_C
  = \quad
  \mkern-5mu
  \begin{tikzpicture}[baseline=-0.65ex]
    \matrix[
    matrix of math nodes,
    column sep=.2ex,
    left delimiter={[},
    right delimiter={]}
    ] (m)
    {
      0   & 0   & 0   & 0 \\
      p_A & p_B & p_C & p_D \\
      0   & 0   & 0   & 0 \\
      p_A & p_B & 0   & 0 \\
    };
    \draw[dashed, ultra thin]
    ([xshift=4.8ex]m-1-1.north east) -- ([xshift=4.2ex]m-4-1.south east);
    \draw[dashed, ultra thin]
    (m-2-1.south west) -- (m-2-4.south east);
    \node[above,text depth=1pt] at (m-1-1.north) {$\scriptstyle A$};
    \node[above,text depth=1pt] at (m-1-2.north) {$\scriptstyle B$};
    \node[above,text depth=1pt] at (m-1-3.north) {$\scriptstyle C$};
    \node[above,text depth=1pt] at (m-1-4.north) {$\scriptstyle D$};
    \node[left,overlay] at ([xshift=-2.3ex]m-1-1.west) {$\scriptstyle A$};
    \node[left,overlay] at ([xshift=-1.6ex]m-2-1.west) {$\scriptstyle B$};
    \node[left,overlay] at ([xshift=-2.3ex]m-3-1.west) {$\scriptstyle C$};
    \node[left,overlay] at ([xshift=-1.6ex]m-4-1.west) {$\scriptstyle D$};
  \end{tikzpicture}\mkern-5mu
  \label{eq:t1t2matrix}
\end{equation}
%

%
We define the mean of the initial generation as $E[G_0] = (p_A, p_B, 0, 0)$---we
assing $p_C = p_D = 0$ since we choose to start by generating a value of type
\ensuremath{T_{1}}.
With $M_C$ and $E[G_0]$ in place, we can apply the equations explained through
Section \ref{sec:bp2} to predict the expected number of \ensuremath{\Conid{A}}, \ensuremath{\Conid{B}}, \ensuremath{\Conid{C}} and \ensuremath{\Conid{D}}
constructors.
%

While this approach works, it completely ignores the types \ensuremath{T_{1}} and \ensuremath{T_{2}} when
calculating $M_C$!
For a large set of mutually-recursive data types involving a large number of
constructors, handling $M_C$ like this results in a high computational cost.
We show next how we cannot only shrink this mean matrix of constructors but also
compute it automatically by making use of data type definitions.

\paragraph{Mean matrix of types}
%
%
If we analyze the mean matrices of \ensuremath{\Conid{Tree'}} (\ref{eq:treepmatrix}) and the
mutually-recursive types \ensuremath{T_{1}} and \ensuremath{T_{2}} (\ref{eq:t1t2matrix}),
%
%
it seems that determining the expected number of offspring generated by a
non-terminal constructor requires us to \emph{count the number of occurrences in
  the ADT which the offspring belongs~to}.
For instance, $m_{NodeA,Leaf}$
%
%
is $2\cdot p_{Leaf}$ (\ref{mNodeLeaf}), where $2$ is the number of occurrences of
\ensuremath{\Conid{Tree'}} in the declaration of \ensuremath{Node_A}. 
Similarly, $m_{BD}$
%
%
is $1 \cdot p_D$, where $1$ is the number of occurrences of \ensuremath{T_{2}} in the
declaration of \ensuremath{\Conid{B}}. 
This observation means that instead of dealing with constructors, we could
directly deal with types!

We can think about a branching process as generating ``place holders'' for
constructors, where place holders can only be populated by constructors of a
certain type.

\begin{figure}[b]
  \vspace{-10pt}
  \begin{framed}
    \begin{tikzpicture}
      [ level/.style={level distance=1.3cm, sibling distance=20mm/#1}]
      \node [circle,draw] (Tree2N) {\ensuremath{Tree'}}
      child {node [rectangle,draw] (Tree2N2) {\ensuremath{Tree'}\ \ \ensuremath{Tree'}}
        edge from parent[draw=none]
          node[left, xshift=-5pt, yshift=5pt] {$p_{NodeA}$}
      }
      child {node [rectangle,draw] (Tree2N3) {\ensuremath{Tree'}}
        edge from parent[draw=none, xshift=10pt]
          node[right, xshift=5pt, yshift=5pt] {$p_{NodeB}$}
      };
      \draw[->] (Tree2N) -- (Tree2N2);
      \draw[->] (Tree2N) -- (Tree2N3);
      \draw[dashed] (Tree2N2.north) -- (Tree2N2.south);
    \end{tikzpicture}
    \hspace{20pt}
    \begin{tikzpicture}
      [ level/.style={level distance=1.2cm, sibling distance=10mm/#1}]
      \node [circle,draw] (T1) {\ensuremath{T_{1}}}
      child {node [rectangle,draw] (T1T2) {\ensuremath{T_{1}}\ \ \ensuremath{T_{2}}}
        edge from parent[->] node[left] {$\ p_{B}$}
      };
      \draw[dashed] (T1T2.north) -- (T1T2.south);
      \node [right = 0.4cm of T1, circle,draw] (T2){\ensuremath{T_{2}}}
      child {node [rectangle,draw] (T1) {\ensuremath{T_{1}}}
        edge from parent[->] node[right] {$\ p_{D}$}
      };
      \draw[dashed, gray] (-0.75,-1.75) rectangle (1.9,0.65);
    \end{tikzpicture}
  \caption{\label{fig:mean:T} Offspring as types}
  \vspace{-5pt}
\end{framed}
\end{figure}
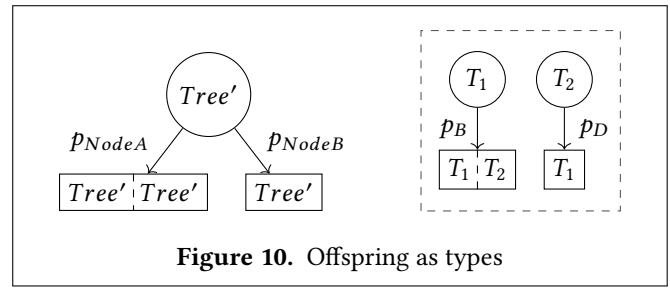

%
Figure \ref{fig:mean:T} illustrates offspring as types for the definitions
\ensuremath{T_{1}}, \ensuremath{T_{2}}, and \ensuremath{\Conid{Tree'}}. 
A place holder of type \ensuremath{T_{1}} can generate a place holder for type \ensuremath{T_{1}}
and a place holder for type \ensuremath{T_{2}}.
A place holder of type \ensuremath{T_{2}} can generate a place holder of type \ensuremath{T_{1}}.
A place holder of type \ensuremath{\Conid{Tree'}} can generate two place holders of type \ensuremath{\Conid{Tree'}}
when generating \ensuremath{Node_A}, one place holder when generating \ensuremath{Node_B}, or zero place
holders when generating a \ensuremath{\Conid{Leaf}} (this last case is not shown in the figure since it
is void).
With these considerations, the mean matrices of types for \ensuremath{\Conid{Tree'}}, written
$M_{Tree'}$; and types \ensuremath{T_{1}} and \ensuremath{T_{2}}, written $M_{T_1T_2}$ are defined as
follows:
\vspace{-5pt}
\begin{equation}
  M_{Tree'}\!
  =\qquad\;
  \mkern-5mu
  \begin{tikzpicture}[
    baseline=-0.65ex,
    every left delimiter/.style={xshift=.5em},
    every right delimiter/.style={xshift=-.5em},
    ]
    \matrix[
    matrix of math nodes,
    inner sep=0pt, 
    nodes={inner sep=.3333em}, 
    left delimiter={[},
    right delimiter={]}
    ] (m)
    {
      2 \cdot p_{NodeA} + p_{NodeB}\\
    };
    \node[above,text depth=1pt] at (m-1-1.north) {$\scriptstyle Tree'$};
    \node[left,overlay] at ([xshift=0.3em]m-1-1.west) {$\scriptstyle Tree'$};
  \end{tikzpicture}\mkern-5mu
  \hspace{20pt}
  M_{T_1T_2}\!
  %
  =\quad\;
  \mkern-5mu
  \begin{tikzpicture}[
    baseline=-0.65ex,
    every left delimiter/.style={xshift=.5em},
    every right delimiter/.style={xshift=-.5em},
    ]
    \matrix[
    matrix of math nodes,
    inner sep=0pt, 
    nodes={inner sep=.3333em}, 
    left delimiter={[},
    right delimiter={]}
    ] (m)
    {
      p_B & p_B\\
      p_D & 0\\
    };
    \node[above,text depth=1pt] at (m-1-1.north) {$\scriptstyle T1$};
    \node[above,text depth=1pt] at (m-1-2.north) {$\scriptstyle T2$};
    \node[left,overlay] at ([xshift=0pt]m-1-1.west) {$\scriptstyle T1$};
    \node[left,overlay] at ([xshift=0pt]m-2-1.west) {$\scriptstyle T2$};
  \end{tikzpicture}\mkern-5mu
  \nonumber
\end{equation}
%
Note how $M_{Tree'}$ shows that
%
the mean matrices of types might reduce a multi-type branching process to a
simple-type one.
%
%

Having the type matrix in place, we can use the following equation (formally
stated and proved in the%
\ifbool{EXTENDED}{
  Appendix \ref{sec:appendixA}%
}{
  supplementary material%
})
to soundly predict the expected number of constructors of a given set of
(possibly) mutually recursive types:
%
%
%
\begin{equation}
(E[G_n^C]).C^t_i = (E[G_n^T]).T_t \cdot p_{C^t_i} \tag{$\forall n \ge 0$}
\end{equation}
%
%
%
Where $G_n^C$ and $G_n^T$ denotes the $n$th-generations of constructors and type
place holders respectively.
$C^t_i$ represents the $i$th-constructor of the type $T_t$.
The equation establishes that, the expected number of constructors $C^t_i$ at
generation $n$ consists of the expected number of type place holders of its type
(i.e., $T_t$) at generation $n$ times the probability of generating that
constructor.
This equation allows us to simplify many of our calculations above by simply
using the mean matrix for types instead of the mean matrix for constructors.

\subsection{Composite types}
In this subsection, we extend our approach in a \emph{modular} manner to deal with
composite ADTs, i.e., ADTs which use already defined types in their
constructors' arguments and which are not involved in the branching process.
%
%
%
%
%
We start by considering the ADT \ensuremath{\Conid{Tree}} modified to carry booleans at the leaves:
%
%
\begin{hscode}\SaveRestoreHook
\column{B}{@{}>{\hspre}l<{\hspost}@{}}%
\column{E}{@{}>{\hspre}l<{\hspost}@{}}%
\>[B]{}\mathbf{data}\;\Conid{Tree}\mathrel{=}Leaf_A\;\Conid{Bool}\mid Leaf_B\;\Conid{Bool}\;\Conid{Bool}\mid \cdots\ {}\<[E]%
\ColumnHook
\end{hscode}\resethooks
Where \ensuremath{\cdots\ } denotes the constructors that remain unmodified.
To predict the expected number of \ensuremath{\Conid{True}} (and analogously of \ensuremath{\Conid{False}})
%
%
constructors, we calculate the multi-type branching process for \ensuremath{\Conid{Tree}} and
multiply each expected number of leaves by the number of arguments of type
\ensuremath{\Conid{Bool}} present in each one:
%
%
%
\begin{align*}
  E[\ensuremath{\Conid{True}}]
  = p_{True}
  \cdot (
  \underbrace{1 \cdot E[\ensuremath{Leaf_A}]}_\text{\text{case}\ \ensuremath{Leaf_A}} +
  \underbrace{2 \cdot E[\ensuremath{Leaf_B}]}_\text{\text{case}\ \ensuremath{Leaf_B}} )
\end{align*}
In this case, \ensuremath{\Conid{Bool}} is a ground type like \ensuremath{\Conid{Int}}, \ensuremath{\Conid{Float}}, etc.
Predictions become more interesting when considering richer composite types
involving, for instance, instantiations of polymorphic types.
%
%
To illustrate this point, consider a modified version of \ensuremath{\Conid{Tree}} where \ensuremath{Leaf_A}
now carries a value of type \ensuremath{\Conid{Maybe}\;\Conid{Bool}}:
%
%
\begin{hscode}\SaveRestoreHook
\column{B}{@{}>{\hspre}l<{\hspost}@{}}%
\column{E}{@{}>{\hspre}l<{\hspost}@{}}%
\>[B]{}\mathbf{data}\;\Conid{Tree}\mathrel{=}Leaf_A\;(\Conid{Maybe}\;\Conid{Bool})\mid Leaf_B\;\Conid{Bool}\;\Conid{Bool}\mid \cdots\ {}\<[E]%
\ColumnHook
\end{hscode}\resethooks
%
%
In order to calculate the expected number of \ensuremath{\Conid{True}}s, now we need to consider
the cases that a value of type \ensuremath{\Conid{Maybe}\;\Conid{Bool}} \emph{actually carries a boolean
  value}, i.e., when a \ensuremath{\Conid{Just}} constructor gets generated:
\begin{align*}
  E[\ensuremath{\Conid{True}}]
  = p_{True} \cdot (1 \cdot E[\ensuremath{Leaf_A}] \cdot p_{Just} + 2 \cdot E[\ensuremath{Leaf_B}])
\end{align*}

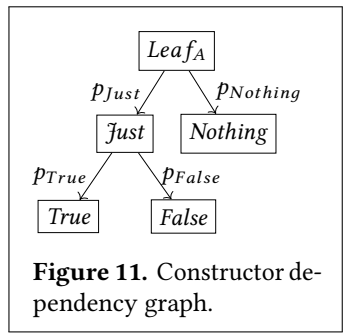
\begin{wrapfigure}{r}{0.25\textwidth}
  \vspace{-13pt}
  \begin{framed}
    \resizebox {4cm} {!} {
      \hspace{-8pt}
    \begin{tikzpicture}
      [level/.style={sibling distance=15mm/#1, level distance=12mm}]
      \node [rectangle,draw] (LeafC){\ensuremath{Leaf_A}}
      child {
        node [rectangle,draw] (Just) {\ensuremath{\Conid{Just}}}
        child {
          node [rectangle,draw,left] (True) {\ensuremath{\Conid{True}}}
          edge from parent[draw=none] node[left] {$p_{True}$}}
        child {
          node [rectangle,draw,right] (False) {\ensuremath{\Conid{False}}}
          edge from parent[draw=none] node[right] {$p_{False}$}}
        edge from parent[draw=none] node[left] {$p_{Just}$}
      }
      child {
        node [rectangle,draw] (Nothing) {\ensuremath{\Conid{Nothing}}}
        edge from parent[draw=none] node[right] {$\ p_{Nothing}$}
      };
      \path[->] (LeafC) edge node {} (Just);
      \path[->] (LeafC) edge node {} (Nothing);
      \path[->] (Just) edge node {} (True);
      \path[->] (Just) edge node {} (False);
    \end{tikzpicture}
    \hspace{5pt}
  }
  \vspace{-15pt}
  \caption{\label{fig:graphdeps} Constructor dependency graph.}
  \vspace{-5pt}
  \end{framed}
  \vspace{-10pt}
\end{wrapfigure}
In the general case, for constructor arguments utilizing other ADTs, it is
necessary to know the chain of constructors required to generate ``foreign''
values---in our example, a \ensuremath{\Conid{True}} value gets generated if a \ensuremath{Leaf_A} gets
generated with a \ensuremath{\Conid{Just}} constructor ``in between.''
To obtain such information, we create of a \emph{constructor dependency graph}
(CDG), that is, a directed graph where each node represents a constructor and
each edge represents its dependency. Each edge is labeled with its corresponding
generation probability.
Figure \ref{fig:graphdeps} shows the CDG for \ensuremath{\Conid{Tree}} starting from the \ensuremath{Leaf_A}
constructor.
Having this graph together with the application of the multi-type branching
process, we can predict the expected number of constructors belonging to
external ADTs.
It is enough to multiply the probabilities at each edge of the path between
every constructor involved in the branching process and the desired
external~constructor.

The extensions described so far enable our tool (presented in the next section)
to make predictions about \quickcheck generators for ADTs defined in many
existing Haskell libraries.



\section{Implementation}
\label{sec:implementation}

%
%
%
\dragen is a tool chain written in Haskell that implements the multi-type
branching processes (Section \ref{sec:bp2} and \ref{sec:terminals}) and its
extensions (Section \ref{sec:opt}) together with a distribution optimizer, which
calibrates the probabilities involved in generators to fit developers' demands.
%
%
\dragen synthesizes generators by calling the Template Haskell function
\ensuremath{\Varid{dragenArbitrary}\mathbin{::}\Conid{Name}\to \Conid{Size}\to \Conid{CostFunction}\to \Conid{Q}\;[\mskip1.5mu \Conid{Dec}\mskip1.5mu]}, where developers
indicate the target ADT for which they want to obtain a \quickcheck generator;
the desired generation size, needed by our prediction mechanism in order to
calculate the distribution at the last generation level; and a \emph{cost
  function} encoding the desired generation distribution.

The design decision to use a probability optimizer rather than search for an
analytical solution is driven by two important aspects of the problem we aim to
solve.
%
Firstly, the computational cost of exactly solving a non-linear system of
equations (such as those arising from branching processes) can be prohibitively
high when dealing with a large number of constructors, thus a large number of
unknowns to be solved for.
Secondly, the existence of such exact solutions is not guaranteed due to the
implicit invariants the data types under consideration might have.
%
%
In such cases, we believe it is much more useful to construct a distribution
that approximates the user's goal, than to abort the entire compilation process.
We give an example of this approximate solution finding behavior later in this
section.
%
%

%
%

\vspace{5pt}
\subsection{Cost functions}
The optimization process is guided by a user-provided cost function.
%
%
In our setting, a cost function assigns a real number (a cost) to the
combination of a generation size (chosen by the user) and a mapping from
constructors to probabilities:
\begin{hscode}\SaveRestoreHook
\column{B}{@{}>{\hspre}l<{\hspost}@{}}%
\column{3}{@{}>{\hspre}l<{\hspost}@{}}%
\column{E}{@{}>{\hspre}l<{\hspost}@{}}%
\>[3]{}\mathbf{type}\;\Conid{CostFunction}\mathrel{=}\Conid{Size}\to \Conid{ProbMap}\to \Conid{Double}{}\<[E]%
\ColumnHook
\end{hscode}\resethooks
%
%
%
Type \ensuremath{\Conid{ProbMap}} encodes the mapping from constructor names to real numbers.
Our optimization algorithm works by generating several \ensuremath{\Conid{ProbMap}} candidates that
are evaluated through the provided cost function in order to choose the most
suitable one.
Cost functions are expected to return a smaller positive number as the predicted
distribution obtained from its parameters gets closer to a certain \emph{target
  distribution}, which depends on what property that particular cost function is
intended to encode.
Then, the optimizator simply finds the best \ensuremath{\Conid{ProbMap}} by minimizing the provided
cost function.

Currently, our tool provides a basic set of cost functions to easily describe
the expected distribution of the derived generator.
%
%
%
%
For instance, \ensuremath{\Varid{uniform}\mathbin{::}\Conid{CostFunction}} encodes constuctor-wise uniform
generation, an interesting property that naturally arises from our generation
process formalization.
%
%
It guides the optimization process to a generation distribution that minimizes
the difference between the expected number of each generated constructor and the
generation size.
%
%
Moreover, the user can restrict the generation distribution to a certain subset
of constructors using the cost functions \ensuremath{\Varid{only}\mathbin{::}[\mskip1.5mu \Conid{Name}\mskip1.5mu]\to \Conid{CostFunction}} and
\ensuremath{\Varid{without}\mathbin{::}[\mskip1.5mu \Conid{Name}\mskip1.5mu]\to \Conid{CostFunction}} to describe these restrictions.
In this case, the whitelisted constructors are then generated following the
\ensuremath{\Varid{uniform}} behavior.
Similarly, if the branching process involves mutually recursive data types, the
user could restrict the generation to a certain subset of data types by using
the functions \ensuremath{\Varid{onlyTypes}} and \ensuremath{\Varid{withoutTypes}}.
Additionally, when the user wants to generate constructors according to certain
proportions, \ensuremath{\Varid{weighted}\mathbin{::}[\mskip1.5mu (\Conid{Name},\Conid{Int})\mskip1.5mu]\to \Conid{CostFunction}} allows to encode this
property, e.g. three times more \ensuremath{Leaf_A}'s than \ensuremath{Leaf_B}'s.
%
%
%

\begin{table*}[t]
  \caption{Predicted and actual distributions for \ensuremath{\Conid{Tree}} generators using
    different cost functions.}
  \vspace{-5pt}
  \begin{center}
    \begin{tabular}{l|c c c c|c c c c}
      \toprule 
      \textbf{Cost Function}
      & \multicolumn{4}{c|}{\textbf{Predicted Expectation}}
      & \multicolumn{4}{c}{\textbf{Observed Expectation}} \\
      & \ensuremath{Leaf_A} & \ensuremath{Leaf_B} & \ensuremath{Leaf_C} & \ensuremath{\Conid{Node}}
      & \ensuremath{Leaf_A} & \ensuremath{Leaf_B} & \ensuremath{Leaf_C} & \ensuremath{\Conid{Node}} \\
      \midrule 
      \ensuremath{\Varid{uniform}}
      \hfill & 5.26 & 5.26 & 5.21 & 14.73 & 5.27 & 5.26 & 5.21 & 14.74 \\
      \ensuremath{\Varid{weighted}\;[\mskip1.5mu (\textquotesingle Leaf_A,\mathrm{3}),(\textquotesingle Leaf_B,\mathrm{1}),(\textquotesingle Leaf_C,\mathrm{1})\mskip1.5mu]}
      \hfill & 30.07 & 9.76 & 10.15 & 48.96 & 30.06 & 9.75 & 10.16 & 48.98 \\
      \ensuremath{\Varid{weighted}\;[\mskip1.5mu (\textquotesingle Leaf_A,\mathrm{1}),(\textquotesingle Node,\mathrm{3})\mskip1.5mu]}
      \hfill & 10.07 & 3.15 & 17.57 & 29.80 & 10.08 & 3.15 & 17.58 & 29.82 \\
      \ensuremath{\Varid{only}\;[\mskip1.5mu \textquotesingle Leaf_A,\textquotesingle Node\mskip1.5mu]}
      \hfill & 10.41 & 0 & 0 & 9.41 & 10.43 & 0 & 0 & 9.43 \\
      \ensuremath{\Varid{without}\;[\mskip1.5mu \textquotesingle Leaf_C\mskip1.5mu]}
      \hfill & 6.95 & 6.95 & 0 & 12.91 & 6.93 & 6.92 & 0 & 12.86 \\
      \bottomrule 
    \end{tabular}
  \end{center}
  \label{tab:costfunctions}
  \vspace{-5pt}
\end{table*}
%
%

Table \ref{tab:costfunctions} shows the number of expected and observed
constructors of
%
%
different \ensuremath{\Conid{Tree}} generators obtained by using different cost functions.
The observed expectations were calculated averaging the number of constructors
across 100000 generated values.
%
Firstly, note how the generated distributions are soundly predicted by our tool.
In our tests, the small differences between predictions and actual values
dissapear as we increase the number of generated values.
As for the cost functions' behavior, there are some interesting aspects to note.
For instance, in the \ensuremath{\Varid{uniform}} case the optimizer cannot do anything to break
the implicit invariant of the data type:
every binary tree with $n$ nodes has $n+1$ leaves.
Instead, it converges to a solution that ``approximates'' a uniform distribution
around the generation size parameter.
We believe this is desirable behavior, to find an approximate solution when
certain invariants prevent the optimization process from finding an exact
solution.
This way the user does not have to be aware of the possible invariants that the
target data type may have, obtaining a solution that is good enough for most
purposes.
On the other hand, notice that in the \ensuremath{\Varid{weighted}} case at the second row of Table
\ref{tab:costfunctions}, the expected number of generated \ensuremath{\Conid{Node}}s is
considerably large.
This constructor is not listed in the proportions list, hence the optimizer can
freely adjust its probability to satisfy the proportions specified for the
leaves.
%

\subsection{Derivation Process}
%
%
\dragen's derivation process starts at compile-time with a type reification
stage that extracts information about the structure of the types under
consideration.
%
%
It follows an intermediate stage composed of the optimizer for probabilities
used in generators, which is guided by our multi-type branching process model,
parametrized on the cost function provided.
%
%
This optimizer is based on a standard local-search optimization algorithm that
recursively chooses the best mapping from constructors to probabilities in the
current neighborhood.
Neighbors are \ensuremath{\Conid{ProbMap}}s, determined by individually varying the probabilities
for \emph{each constructor} with a predetermined $\Delta$.
%
%
%
Then, to determine the ``best'' probabilities, the local-search applies our
prediction mechanishm to the immediate neighbors that have not yet been visited
by evaluating the cost function to select the most suitable next candidate.
This process continues until a local minimum is reached when there are no new
neighbors to evaluate, or if each step improvement is lower than a minimum
predetermined $\varepsilon$.
%

The final stage synthesizes a \ensuremath{\Conid{Arbitrary}} type-class instance for the target
data types using the optimized generation probabilities.
For this stage, we extend some functionality present in \megadeth in order to
derive generators parametrized by our previously optimized probabilities.
Refer to%
\ifbool{EXTENDED}{
  Appendix \ref{app:implementation}%
}{
  the supplementary material%
}
for further details on the cost functions and algorithms addressed by this
section.

\section{Case Studies}
\label{sec:casestudies}

We start by comparing the generators for the ADT \ensuremath{\Conid{Tree}} derived by \megadeth and
\feat, presented in Section \ref{sec:QC}, with the corresponding generator
derived by \dragen using a \ensuremath{\Varid{uniform}} cost function.
%
%
We used a generation size of 10 both for \megadeth and \dragen, and a generation
size of 400 for \feat---that is, \feat will generate test cases of maximum 400
constructors, since this is the maximum number of constructors generated by our
tool using the generation size cited above.
%
%
Figure \ref{fig:conscount} shows the differences between the complexity of the
generated values in terms of the number of constructors.
%
%
As shown in Figure \ref{fig:tree_megadeth_feat}, generators derived by
\megadeth and \feat produce very narrow distributions, being unable to generate
a diverse variety of values of different sizes.
%
%
In contrast, the \dragen optimized generator provides a much wider distribution,
i.e., from smaller to bigger values.
%
%

\begin{figure}[b] 
  \vspace{-10pt}
  \includegraphics[width=0.95\columnwidth]{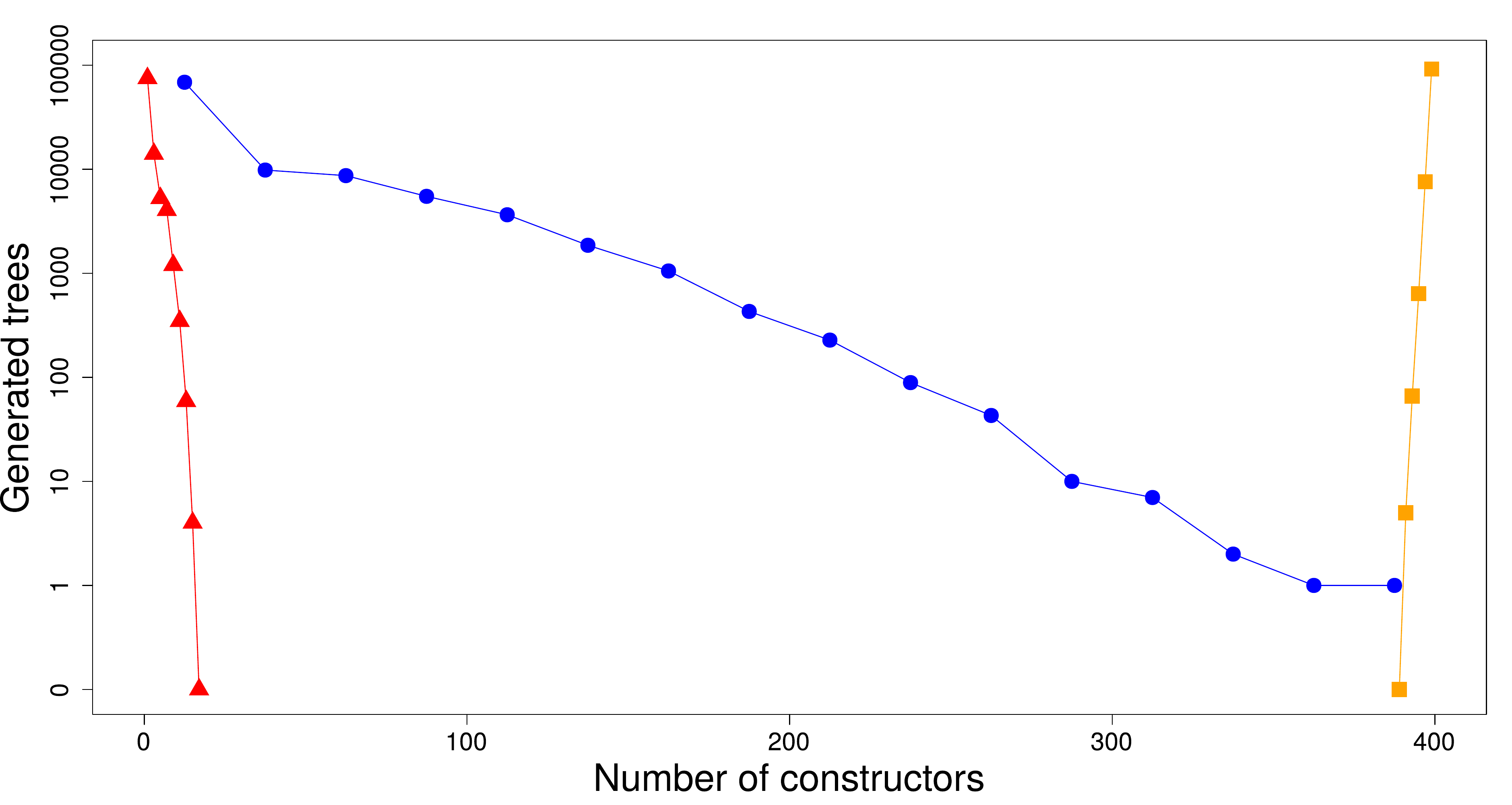}
  \vspace{-10pt}
  \caption{ \megadeth ({\color{red} $\blacktriangle$}) vs. \feat
    ({\color[RGB]{255, 163, 0} {\tiny $\blacksquare$}}) vs. \dragen
    ({\color{blue} $\bullet$}) generated distributions for type \ensuremath{\Conid{Tree}}.}
  \label{fig:conscount}
\end{figure}

%
It is likely that the richer the values generated, the better the chances of
covering more code, and thus of finding more bugs.
The next case studies provide evidence in that direction.

Although \dragen can be used to test Haskell code, we follow the same philosophy
as \quickfuzz, targeting three complex and widely used external programs to
evaluate how well our derived generators behave.
%
%
%
These applications are \emph{GNU bash 4.4}---a widely used Unix shell, \emph{GNU
  CLISP 2.49}---the GNU Common Lisp compiler, and \emph{giffix}---a small test
utility from the \emph{GIFLIB 5.1} library focused on reading and writing Gif
images.
It is worth noticing that these applications are not written in Haskell.
%
%
Nevertheless, there are Haskell libraries designed to inter-operate with them:
\emph{language-bash}, \emph{atto-lisp}, and \emph{JuicyPixels}, respectively.
These libraries provide ADT definitions which we used to synthesize \dragen
generators for the inputs of the aforementioned applications.
Moreover, they also come with serialization functions that allow us to transform
the randomly generated Haskell values
%
%
into the actual test files that we used to test each external program.
The case studies contain mutually recursive and composite ADTs with a wide
number of constructors (e.g., GNU bash spans 31 different ADTs and 136 different
constuctors)---refer to%
\ifbool{EXTENDED}{
  \ref{app:casestudies}%
}{
  the supplementary material%
}
for a rough estimation of the
scale of such data types and the data types involved with them.
%

For our experiments, we use the coverage measure known as \emph{execution path}
employed by American Fuzzy Lop (AFL) \cite{afl}---a well known fuzzer.
%
%
%
It was chosen in this work since it is also used in the work by
\citeauthor{grieco2017} \cite{grieco2017} to compare \megadeth with other
techniques.
The process consists of the \emph{instrumentation} of the binaries under test,
making them able to return the path in the code taken by each execution.
%
%
Then, we use AFL to count how many different executions are triggered by a set
of randomly generated files---also known as a corpus.
In this evaluation, we compare how different \quickcheck generators, derived
using \megadeth and using our approach, result in different code coverage when
testing external programs, as a function of the size of a set of independently,
randomly generated corpora.
We have not been able to automatically derive such generators using \feat, since
it does not work with some Haskell extensions used in the bridging libraries.
%
%
%

We generated each corpus using the same ADTs and generation sizes for each
derivation mechanism.
We used a generation size of 10 for CLISP and bash files, and a size of 5 for
Gif files.
For \dragen, we used \ensuremath{\Varid{uniform}} cost functions to reduce any external bias.
In this manner, any observed difference in the code coverage triggered by the
corpora generated using each derivation mechanism is entirely caused by the
optimization stage that our predictive approach performs, which does not
represent an extra effort for the programmer.
%
%
Moreover, we repeat each experiment 30 times using independently generated
corpora for each combination of derivation mechanism and corpus size.
%

Figure \ref{fig:all} compares the mean number of different execution paths
triggered by each pair of generators and corpus sizes, with error bars
indicating 95\% confidence intervals of the mean.
\begin{figure*}[t]
  \includegraphics[width=\textwidth]{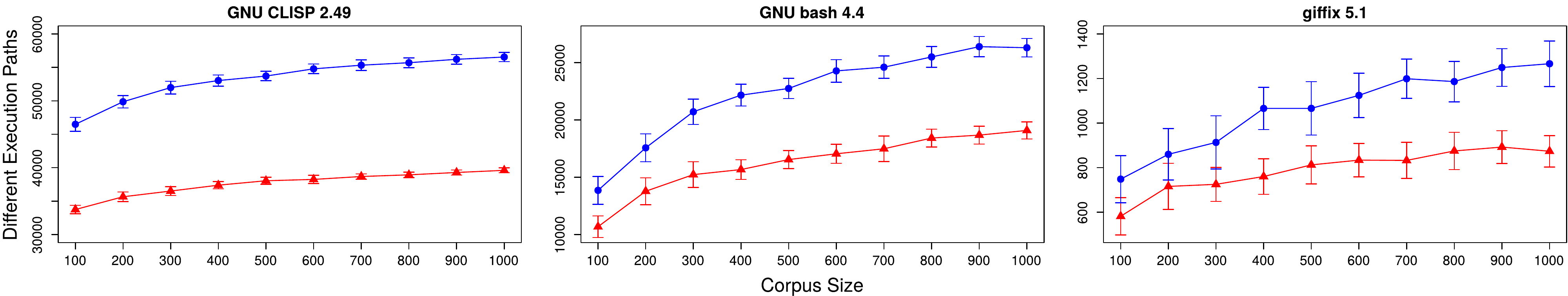}
  \vspace{-20pt}
  \caption{\label{fig:all} Path coverage comparison between \megadeth
    ({\color{red} $\blacktriangle$}) and \dragen ({\color{blue} $\bullet$}).}
  \vspace{-5pt}
\end{figure*}
%
%
It is easy to see how the \dragen generators synthesize test cases capable of
triggering a much larger number of different execution paths in comparison to
\megadeth ones.
Our results indicate average increases approximately between $35\%$ and $41\%$
with an standard error close to $0.35\%$ in the number of different execution
paths triggered in the programs under test.

An attentive reader might remember that \megadeth tends to derive generators
which produce very small test cases.
If we consider that small test cases should take less time (on average) to be
tested, is fair to think there is a trade-off between being able to test a
bigger number of smaller test cases or a smaller number of bigger ones having
the same time available.
However, when testing external software like in our experiments, it is important
to consider the time overhead introduced by the operating system.
%
%
In this scenario, it is much more preferable to test interesting values over smaller
ones.
%
%
In our tests, size differences between the generated values of each tool does do
not result in significant differences in the runtimes required to test each
corpora---refer to%
\ifbool{EXTENDED}{
  Appendix \ref{app:casestudies}.%
}{
  the supplementary material for further details.%
}
%
%
%
A user is most likely to get better results by using our tool instead of
\megadeth, with virtually \emph{the same effort.}
\looseness=-1

We also remark that, if we run sufficiently many tests, then the expected code
coverage will tend towards 100\% of the reachable code in both cases.
%
%
However, in practice, our approach is more likely to achieve higher code
coverage for the same number of test cases.



\section{Related Work}
\label{sec:related}

%

Fuzzers are tools to tests programs against randomly generated unexpected
inputs.
%
%
\quickfuzz \cite{GriecoCB16, grieco2017} is a tool that synthesizes data with
rich structure, that is, well-typed files which can be used as initial ``seeds''
for state-of-the-art fuzzers---a work flow which discovered many unknown
vulnerabilities.
%
%
%
Our work could help to improve the variation of the generated initial seeds, by
varying the distribution of \quickfuzz generators---an interesting direction for
future work.


\emph{SmallCheck} \cite{RuncimanNL08} provides a framework to exhaustively test
data sets up to a certain (small) size.
%
%
The authors also propose a variation called \emph{Lazy SmallCheck}, which
avoids the generation of multiple variants which are passed to the test, but not
actually used.

\quickcheck has been used to generate well-typed lambda terms in order to test
compilers \cite{Palka11}.
%
%
Recently, \citeauthor{MidtgaardJKNN17} extend such a technique to test compilers
for impure programming languages \cite{MidtgaardJKNN17}.
%

%
\emph{Luck} \cite{LampropoulosGHH17} is a domain specific language for writing
testing properties and \quickcheck generators at the same time.
%
%
We see \emph{Luck}'s approach as orthogonal to ours, which is mostly intended to
be used when we do not know any specific property of the system under test,
although we consider that borrowing some functionality from \emph{Luck} into
\dragen is an interesting path for future work.

%
%
Recently, \citeauthor{Lampropoulos2017} propose a framework to automatically
derive random generators for a large subclass of Coqs' inductively defined
relations \cite{Lampropoulos2017}.
%
%
This derivation process also provides proof terms certifying that each derived
generator is sound and complete with respect to the inductive relation it was
derived from.

%

\emph{Boltzmann models} \cite{Duchon2004} are a general approach to randomly
generating combinatorial structures such as trees and graphs---also extended to
work with closed simply-typed lambda terms \cite{Bendkowski2017}.
By implementing a \emph{Boltzmann sampler}, it is possible to obtain a random
generator built around such models which uniformly generates values of a target
size with a certain size tolerance.
%
%
%
%
However, this approach has practical limitations.
Firstly, the framework is not expressive enough to represent complex constrained
data structures, e.g red-black trees.
Secondly, Boltzmann samplers give the user no control over the distribution of
generated values besides ensuring size-uniform generation.
They work well in theory but further work is required to apply them to complex
structures \cite{Poulding2017}.
Conversely, \dragen provides a simple mechanism to predict and tune the overall
distribution of constructors \emph{analytically at compile-time}, using
statically known type information, and requiring no runtime reinforcements to
ensure the predicted distributions.
Future work will explore the connections between branching processes and
Boltzmann models.
\looseness=-1
%
%

Similarly to our work, \citeauthor{feldt2013} propose \godeltest
\cite{feldt2013}, a search-based framework for generating biased data.
%
%
It relies on non-determinism to generate a wide range of data structures, along
with metaheuristic search to optimize the parameters governing the desired
biases in the generated data.
%
%
Rather than using metaheuristic search, our approach employs a completely
analytical process to predict the generation distribution at each optimization
step.
%
%
A strength of the \godeltest approach is that it can optimize the probability
parameters even when there is no specific target distribution over the
constructors---this allows exploiting software behavior under test to guide the
parameter optimization.
%

The efficiency of random testing is improved if the generated inputs are evenly
spread across the input domain \cite{chan1996}.
This is the main idea of \emph{Adaptive Random Testing} (ART) \cite{chen2005}.
%
However, this work only covers the particular case of testing programs with
numerical inputs and it has also been argued that adaptive random testing has
inherent inefficiencies compared to random testing \cite{Arcuri2011}.
This strategy is later extended in \cite{ciupa2008} for object-oriented
programs.
These approaches present no analysis of the distribution obtained by the
 heuristics used, therefore we see them as orthogonal work to ours.



\section{Final Remarks}
\label{sec:conc}
%
%
We discover an interplay between the stochastic theory of branching processes
and algebraic data types structures.
This connection enables us to describe a solid mathematical foundation to
capture the behavior of our derived \quickcheck generators.
Based on our formulas, we implement a heuristic to automatically adjust the
expected number of constructors being generated as a way to control generation
distributions.
%

%
One holy grail in testing is the generation of structured data which fulfills
certain invariants.
%
%
We believe that our work could be used to enforce some invariants on data ``up
to some degree.''
For instance, by inspecting programs' source code, we could extract the
pattern-matching patterns from programs (e.g., \ensuremath{(\Conid{Cons}\;(\Conid{Cons}\;\Varid{x}))}) and derive
generators which ensure that such patterns get exercised a certain amount of
times (on average)---intriguing thoughts to drive our future work.
%

\balance



\begin{acks}
  We would like to thank Micha\l{} Pa\l{}ka, Nick Smallbone, Martin Ceresa and
  Gustavo Grieco for comments on an early draft.
  This work was funded by the Swedish Foundation for Strategic Research (SSF)
  under the project Octopi (Ref. RIT17-0023) and WebSec (Ref. RIT17-0011) as
  well as the Swedish research agency Vetenskapsr{\aa}det.
\end{acks}


\bibliography{local.bib}


\begin{thebibliography}{00}


\ifx \showCODEN    \undefined \def \showCODEN     #1{\unskip}     \fi
\ifx \showDOI      \undefined \def \showDOI       #1{#1}\fi
\ifx \showISBNx    \undefined \def \showISBNx     #1{\unskip}     \fi
\ifx \showISBNxiii \undefined \def \showISBNxiii  #1{\unskip}     \fi
\ifx \showISSN     \undefined \def \showISSN      #1{\unskip}     \fi
\ifx \showLCCN     \undefined \def \showLCCN      #1{\unskip}     \fi
\ifx \shownote     \undefined \def \shownote      #1{#1}          \fi
\ifx \showarticletitle \undefined \def \showarticletitle #1{#1}   \fi
\ifx \showURL      \undefined \def \showURL       {\relax}        \fi
\providecommand\bibfield[2]{#2}
\providecommand\bibinfo[2]{#2}
\providecommand\natexlab[1]{#1}
\providecommand\showeprint[2][]{arXiv:#2}

\bibitem[\protect\citeauthoryear{Arcuri and Briand}{Arcuri and Briand}{2011}]%
        {Arcuri2011}
\bibfield{author}{\bibinfo{person}{A. Arcuri} {and} \bibinfo{person}{L.
  Briand}.} \bibinfo{year}{2011}\natexlab{}.
\newblock \showarticletitle{Adaptive Random Testing: An Illusion of
  Effectiveness?}. In \bibinfo{booktitle}{{\em Proc. of the International
  Symposium on Software Testing and Analysis}} {\em (\bibinfo{series}{ISSTA
  '11})}. \bibinfo{publisher}{ACM}.
\newblock


\bibitem[\protect\citeauthoryear{Arkin, Stender, and McGraw}{Arkin
  et~al\mbox{.}}{2005}]%
        {pentest}
\bibfield{author}{\bibinfo{person}{B. Arkin}, \bibinfo{person}{S. Stender},
  {and} \bibinfo{person}{G. McGraw}.} \bibinfo{year}{2005}\natexlab{}.
\newblock \showarticletitle{Software penetration testing}.
\newblock \bibinfo{journal}{{\em IEEE Security Privacy\/}}
  (\bibinfo{year}{2005}).
\newblock


\bibitem[\protect\citeauthoryear{Arts, Hughes, Norell, and Svensson}{Arts
  et~al\mbox{.}}{2015}]%
        {ArtsHNS15}
\bibfield{author}{\bibinfo{person}{T. Arts}, \bibinfo{person}{J. Hughes},
  \bibinfo{person}{U. Norell}, {and} \bibinfo{person}{H. Svensson}.}
  \bibinfo{year}{2015}\natexlab{}.
\newblock \showarticletitle{Testing {AUTOSAR} software with {QuickCheck}}. In
  \bibinfo{booktitle}{{\em In Proc. of {IEEE} International Conference on
  Software Testing, Verification and Validation, {ICST} Workshops}}.
\newblock


\bibitem[\protect\citeauthoryear{Balakrishnan, Voinov, and
  Nikulin}{Balakrishnan et~al\mbox{.}}{2013}]%
        {chisquarebook}
\bibfield{author}{\bibinfo{person}{N. Balakrishnan}, \bibinfo{person}{V.
  Voinov}, {and} \bibinfo{person}{M.S Nikulin}.}
  \bibinfo{year}{2013}\natexlab{}.
\newblock \bibinfo{booktitle}{{\em {Chi-Squared Goodness of Fit Tests with
  Applications.} 1st Edition}}.
\newblock \bibinfo{publisher}{{Academic Press}}.
\newblock
\showISBNx{9780123971944}


\bibitem[\protect\citeauthoryear{Bendkowski, Grygiel, and Tarau}{Bendkowski
  et~al\mbox{.}}{2017}]%
        {Bendkowski2017}
\bibfield{author}{\bibinfo{person}{M. Bendkowski}, \bibinfo{person}{K.
  Grygiel}, {and} \bibinfo{person}{P. Tarau}.} \bibinfo{year}{2017}\natexlab{}.
\newblock \showarticletitle{Boltzmann Samplers for Closed Simply-Typed Lambda
  Terms}. In \bibinfo{booktitle}{{\em In Proc. of International Symposium on
  Practical Aspects of Declarative Languages}}. \bibinfo{publisher}{ACM}.
\newblock


\bibitem[\protect\citeauthoryear{Chan, Chen, Mak, and Yu}{Chan
  et~al\mbox{.}}{1996}]%
        {chan1996}
\bibfield{author}{\bibinfo{person}{F.T. Chan}, \bibinfo{person}{T.Y. Chen},
  \bibinfo{person}{I.K. Mak}, {and} \bibinfo{person}{Y.T. Yu}.}
  \bibinfo{year}{1996}\natexlab{}.
\newblock \showarticletitle{Proportional sampling strategy: guidelines for
  software testing practitioners}.
\newblock \bibinfo{journal}{{\em Information and Software Technology\/}}
  \bibinfo{volume}{38}, \bibinfo{number}{12} (\bibinfo{year}{1996}),
  \bibinfo{pages}{775 -- 782}.
\newblock


\bibitem[\protect\citeauthoryear{Chen, Leung, and Mak}{Chen
  et~al\mbox{.}}{2005}]%
        {chen2005}
\bibfield{author}{\bibinfo{person}{T.~Y. Chen}, \bibinfo{person}{H. Leung},
  {and} \bibinfo{person}{I.~K. Mak}.} \bibinfo{year}{2005}\natexlab{}.
\newblock \showarticletitle{Adaptive Random Testing}. In
  \bibinfo{booktitle}{{\em Advances in Computer Science - ASIAN 2004.
  Higher-Level Decision Making}}, \bibfield{editor}{\bibinfo{person}{Michael~J.
  Maher}} (Ed.). \bibinfo{publisher}{Springer Berlin Heidelberg}.
\newblock


\bibitem[\protect\citeauthoryear{Ciupa, Leitner, Oriol, and Meyer}{Ciupa
  et~al\mbox{.}}{2008}]%
        {ciupa2008}
\bibfield{author}{\bibinfo{person}{I. Ciupa}, \bibinfo{person}{A. Leitner},
  \bibinfo{person}{M. Oriol}, {and} \bibinfo{person}{B. Meyer}.}
  \bibinfo{year}{2008}\natexlab{}.
\newblock \showarticletitle{ARTOO: adaptive random testing for object-oriented
  software}. In \bibinfo{booktitle}{{\em Proc. of International Conference on
  Software Engineering}}. \bibinfo{publisher}{ACM/IEEE}.
\newblock


\bibitem[\protect\citeauthoryear{Claessen, Dureg{\aa}rd, and Palka}{Claessen
  et~al\mbox{.}}{2014}]%
        {ClaessenDP14}
\bibfield{author}{\bibinfo{person}{K. Claessen}, \bibinfo{person}{J.
  Dureg{\aa}rd}, {and} \bibinfo{person}{M.~H. Palka}.}
  \bibinfo{year}{2014}\natexlab{}.
\newblock \showarticletitle{Generating Constrained Random Data with Uniform
  Distribution}. In \bibinfo{booktitle}{{\em Proc. of the Functional and Logic
  Programming {FLOPS}}}.
\newblock


\bibitem[\protect\citeauthoryear{Claessen and Hughes}{Claessen and
  Hughes}{2000}]%
        {ClaessenH00}
\bibfield{author}{\bibinfo{person}{K. Claessen} {and} \bibinfo{person}{J.
  Hughes}.} \bibinfo{year}{2000}\natexlab{}.
\newblock \showarticletitle{{QuickCheck}: A Lightweight Tool for Random Testing
  of {Haskell} Programs}. In \bibinfo{booktitle}{{\em Proc. of the {ACM}
  {SIGPLAN} International Conference on Functional Programming {(ICFP})}}.
\newblock


\bibitem[\protect\citeauthoryear{Duchon, Flajolet, Louchard, and
  Schaeffer}{Duchon et~al\mbox{.}}{2004}]%
        {Duchon2004}
\bibfield{author}{\bibinfo{person}{P. Duchon}, \bibinfo{person}{P. Flajolet},
  \bibinfo{person}{G. Louchard}, {and} \bibinfo{person}{G. Schaeffer}.}
  \bibinfo{year}{2004}\natexlab{}.
\newblock \showarticletitle{Boltzmann Samplers for the Random Generation of
  Combinatorial Structures}.
\newblock \bibinfo{journal}{{\em Combinatorics, Probability and Computing.\/}}
  \bibinfo{volume}{13} (\bibinfo{year}{2004}).
\newblock


\bibitem[\protect\citeauthoryear{Dureg{\aa}rd, Jansson, and Wang}{Dureg{\aa}rd
  et~al\mbox{.}}{2012}]%
        {DuregardJW12}
\bibfield{author}{\bibinfo{person}{J. Dureg{\aa}rd}, \bibinfo{person}{P.
  Jansson}, {and} \bibinfo{person}{M. Wang}.} \bibinfo{year}{2012}\natexlab{}.
\newblock \showarticletitle{Feat: Functional enumeration of algebraic types}.
  In \bibinfo{booktitle}{{\em Proc. of the {ACM} {SIGPLAN} Symposium on
  Haskell}}.
\newblock


\bibitem[\protect\citeauthoryear{Feldt and Poulding}{Feldt and
  Poulding}{2013}]%
        {feldt2013}
\bibfield{author}{\bibinfo{person}{R. Feldt} {and} \bibinfo{person}{S.
  Poulding}.} \bibinfo{year}{2013}\natexlab{}.
\newblock \showarticletitle{Finding test data with specific properties via
  metaheuristic search}. In \bibinfo{booktitle}{{\em Proc. of International
  Symp. on Software Reliability Engineering (ISSRE)}}.
  \bibinfo{publisher}{IEEE}.
\newblock


\bibitem[\protect\citeauthoryear{Grieco, Ceresa, and Buiras}{Grieco
  et~al\mbox{.}}{2016}]%
        {GriecoCB16}
\bibfield{author}{\bibinfo{person}{G. Grieco}, \bibinfo{person}{M. Ceresa},
  {and} \bibinfo{person}{P. Buiras}.} \bibinfo{year}{2016}\natexlab{}.
\newblock \showarticletitle{{QuickFuzz}: An automatic random fuzzer for common
  file formats}. In \bibinfo{booktitle}{{\em Proc. of the International
  Symposium on {Haskell}}}. \bibinfo{publisher}{ACM}.
\newblock


\bibitem[\protect\citeauthoryear{Grieco, Ceresa, Mista, and Buiras}{Grieco
  et~al\mbox{.}}{2017}]%
        {grieco2017}
\bibfield{author}{\bibinfo{person}{G. Grieco}, \bibinfo{person}{M. Ceresa},
  \bibinfo{person}{A. Mista}, {and} \bibinfo{person}{P. Buiras}.}
  \bibinfo{year}{2017}\natexlab{}.
\newblock \showarticletitle{{QuickFuzz} testing for fun and profit}.
\newblock \bibinfo{journal}{{\em Journal of Systems and Software\/}}
  \bibinfo{volume}{134}, \bibinfo{number}{Supp. C} (\bibinfo{year}{2017}).
\newblock


\bibitem[\protect\citeauthoryear{Haccou, Jagers, and Vatutin}{Haccou
  et~al\mbox{.}}{2005}]%
        {gwbook}
\bibfield{author}{\bibinfo{person}{P. Haccou}, \bibinfo{person}{P. Jagers},
  {and} \bibinfo{person}{V. Vatutin}.} \bibinfo{year}{2005}\natexlab{}.
\newblock \bibinfo{booktitle}{{\em Branching processes. Variation, growth, and
  extinction of populations}}.
\newblock \bibinfo{publisher}{{Cambridge University Press}}.
\newblock
\showISBNx{978-0-521-83220-5}


\bibitem[\protect\citeauthoryear{Hughes, B, Arts, and Norell}{Hughes
  et~al\mbox{.}}{2016a}]%
        {HughesPAN16}
\bibfield{author}{\bibinfo{person}{J. Hughes}, \bibinfo{person}{C.~Pierce B},
  \bibinfo{person}{T. Arts}, {and} \bibinfo{person}{U. Norell}.}
  \bibinfo{year}{2016}\natexlab{a}.
\newblock \showarticletitle{Mysteries of {DropBox}: Property-Based Testing of a
  Distributed Synchronization Service}. In \bibinfo{booktitle}{{\em Proc. of
  the Int. Conference on Software Testing, Verification and Validation,
  {ICST}}}.
\newblock


\bibitem[\protect\citeauthoryear{Hughes, Norell, Smallbone, and Arts}{Hughes
  et~al\mbox{.}}{2016b}]%
        {HughesNSA16}
\bibfield{author}{\bibinfo{person}{J. Hughes}, \bibinfo{person}{U. Norell},
  \bibinfo{person}{N. Smallbone}, {and} \bibinfo{person}{T. Arts}.}
  \bibinfo{year}{2016}\natexlab{b}.
\newblock \showarticletitle{Find more bugs with {QuickCheck}!}. In
  \bibinfo{booktitle}{{\em Proc. of the International Workshop on Automation of
  Software Test, {AST@ICSE}}}.
\newblock


\bibitem[\protect\citeauthoryear{Lampropoulos, Gallois{-}Wong, Hritcu, Hughes,
  Pierce, and Xia}{Lampropoulos et~al\mbox{.}}{2017a}]%
        {LampropoulosGHH17}
\bibfield{author}{\bibinfo{person}{L. Lampropoulos}, \bibinfo{person}{D.
  Gallois{-}Wong}, \bibinfo{person}{C. Hritcu}, \bibinfo{person}{J. Hughes},
  \bibinfo{person}{B.~C. Pierce}, {and} \bibinfo{person}{L. Xia}.}
  \bibinfo{year}{2017}\natexlab{a}.
\newblock \showarticletitle{Beginner's luck: a language for property-based
  generators}. In \bibinfo{booktitle}{{\em Proc. of the {ACM} {SIGPLAN}
  Symposium on Principles of Programming Languages, {POPL}}}.
\newblock


\bibitem[\protect\citeauthoryear{Lampropoulos, Paraskevopoulou, and
  Pierce}{Lampropoulos et~al\mbox{.}}{2017b}]%
        {Lampropoulos2017}
\bibfield{author}{\bibinfo{person}{L. Lampropoulos}, \bibinfo{person}{Z.
  Paraskevopoulou}, {and} \bibinfo{person}{B.~C. Pierce}.}
  \bibinfo{year}{2017}\natexlab{b}.
\newblock \showarticletitle{Generating Good Generators for Inductive
  Relations}.
\newblock \bibinfo{journal}{{\em In Proc. ACM on Programming Languages\/}}
  \bibinfo{volume}{2}, \bibinfo{number}{POPL}, Article \bibinfo{articleno}{45}
  (\bibinfo{year}{2017}).
\newblock


\bibitem[\protect\citeauthoryear{{M. Zalewski}}{{M. Zalewski}}{2010}]%
        {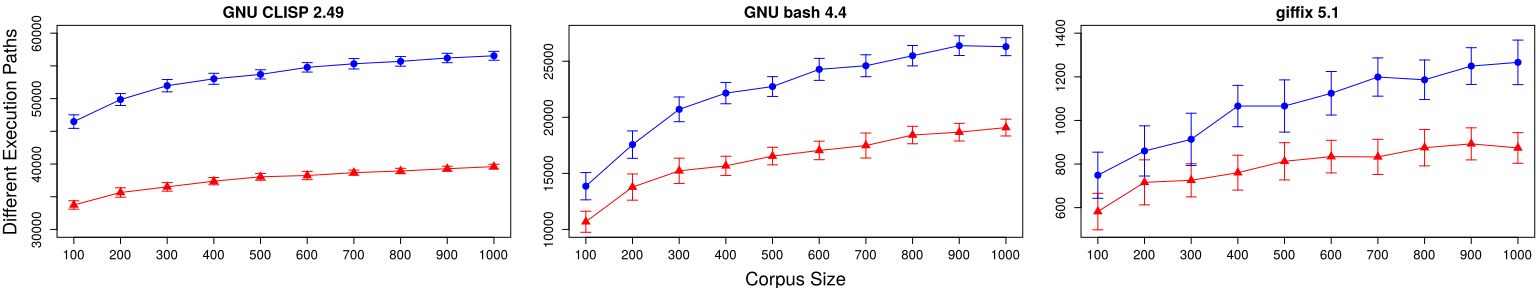}
\bibfield{author}{\bibinfo{person}{{M. Zalewski}}.}
  \bibinfo{year}{2010}\natexlab{}.
\newblock \bibinfo{title}{{American Fuzzy Lop: a security-oriented fuzzer}}.
\newblock \bibinfo{howpublished}{\mbox{\url{http://lcamtuf.coredump.cx/afl/}}}.
    (\bibinfo{year}{2010}).
\newblock


\bibitem[\protect\citeauthoryear{McBride and Paterson}{McBride and
  Paterson}{2008}]%
        {Mcbride2008}
\bibfield{author}{\bibinfo{person}{C. McBride} {and} \bibinfo{person}{R.
  Paterson}.} \bibinfo{year}{2008}\natexlab{}.
\newblock \showarticletitle{Applicative Programming with Effects}.
\newblock \bibinfo{journal}{{\em Journal of Functional Programming\/}}
  \bibinfo{volume}{18}, \bibinfo{number}{1} (\bibinfo{date}{Jan.}
  \bibinfo{year}{2008}).
\newblock


\bibitem[\protect\citeauthoryear{Midtgaard, Justesen, Kasting, Nielson, and
  Nielson}{Midtgaard et~al\mbox{.}}{2017}]%
        {MidtgaardJKNN17}
\bibfield{author}{\bibinfo{person}{J. Midtgaard}, \bibinfo{person}{M.~N.
  Justesen}, \bibinfo{person}{P. Kasting}, \bibinfo{person}{F. Nielson}, {and}
  \bibinfo{person}{H.~R. Nielson}.} \bibinfo{year}{2017}\natexlab{}.
\newblock \showarticletitle{Effect-driven {QuickChecking} of compilers}.
\newblock \bibinfo{journal}{{\em {In Proceedings of the {ACM} on Programming
  Languages, Volume 1}\/}} \bibinfo{number}{{ICFP}} (\bibinfo{year}{2017}).
\newblock


\bibitem[\protect\citeauthoryear{Mitchell}{Mitchell}{2007}]%
        {mitchell2007}
\bibfield{author}{\bibinfo{person}{N. Mitchell}.}
  \bibinfo{year}{2007}\natexlab{}.
\newblock \showarticletitle{Deriving Generic Functions by Example}. In
  \bibinfo{booktitle}{{\em Proc. of the 1st York Doctoral Syposium}}.
  \bibinfo{publisher}{Tech. Report YCS-2007-421, Department of Computer
  Science, University of York, UK}, \bibinfo{pages}{55--62}.
\newblock


\bibitem[\protect\citeauthoryear{Pa{\l}ka, Claessen, Russo, and
  Hughes}{Pa{\l}ka et~al\mbox{.}}{2011}]%
        {Palka11}
\bibfield{author}{\bibinfo{person}{M. Pa{\l}ka}, \bibinfo{person}{K. Claessen},
  \bibinfo{person}{A. Russo}, {and} \bibinfo{person}{J. Hughes}.}
  \bibinfo{year}{2011}\natexlab{}.
\newblock \showarticletitle{Testing and Optimising Compiler by Generating
  Random Lambda Terms}. In \bibinfo{booktitle}{{\em The {IEEE/ACM}
  International Workshop on Automation of Software Test (AST 2011)}}.
\newblock


\bibitem[\protect\citeauthoryear{Poulding and Feldt}{Poulding and
  Feldt}{2017}]%
        {Poulding2017}
\bibfield{author}{\bibinfo{person}{S.~M. Poulding} {and} \bibinfo{person}{R.
  Feldt}.} \bibinfo{year}{2017}\natexlab{}.
\newblock \showarticletitle{Automated Random Testing in Multiple Dispatch
  Languages}.
\newblock \bibinfo{journal}{{\em IEEE International Conference on Software
  Testing, Verification and Validation (ICST)\/}} (\bibinfo{year}{2017}).
\newblock


\bibitem[\protect\citeauthoryear{Runciman, Naylor, and Lindblad}{Runciman
  et~al\mbox{.}}{2008}]%
        {RuncimanNL08}
\bibfield{author}{\bibinfo{person}{C. Runciman}, \bibinfo{person}{M. Naylor},
  {and} \bibinfo{person}{F. Lindblad}.} \bibinfo{year}{2008}\natexlab{}.
\newblock \showarticletitle{{S}mallcheck and {L}azy {S}mallcheck: automatic
  exhaustive testing for small values}. In \bibinfo{booktitle}{{\em Proc. of
  the {ACM} {SIGPLAN} Symposium on Haskell}}.
\newblock


\bibitem[\protect\citeauthoryear{Sheard and Jones}{Sheard and Jones}{2002}]%
        {SheardJ02}
\bibfield{author}{\bibinfo{person}{T. Sheard} {and} \bibinfo{person}{Simon
  L.~Peyton Jones}.} \bibinfo{year}{2002}\natexlab{}.
\newblock \showarticletitle{Template meta-programming for {Haskell}}.
\newblock \bibinfo{journal}{{\em {SIGPLAN} Notices\/}} \bibinfo{volume}{37},
  \bibinfo{number}{12} (\bibinfo{year}{2002}), \bibinfo{pages}{60--75}.
\newblock


\bibitem[\protect\citeauthoryear{Watson and Galton}{Watson and Galton}{1875}]%
        {gw1875}
\bibfield{author}{\bibinfo{person}{H.~W. Watson} {and} \bibinfo{person}{F.
  Galton}.} \bibinfo{year}{1875}\natexlab{}.
\newblock \showarticletitle{On the probability of the extinction of families}.
\newblock \bibinfo{journal}{{\em The Journal of the Anthropological Institute
  of Great Britain and Ireland\/}} (\bibinfo{year}{1875}).
\newblock


\end{thebibliography}


\ifbool{EXTENDED}{
\clearpage \appendix
\section{Demonstrations}
\label{sec:appendixA}

In this appendix, we provide the formal development to show that the mean matrix
of types can be used to predict the distribution of constructors.
We start by defining some terminology.

\begin{helpers-def}
  Let $T_t$ be a data type defined as a sum of type constructors:
  \begin{align*}
    T_t := C_1^t + C_2^t + \cdots + C_n^t
  \end{align*}
  Where each constructor is defined as a product of data types:
  \begin{align*}
    C_c^t := T_1 \times T_2 \times \cdots \times T_m
  \end{align*}
  We will define the following observation functions:
  \begin{align*}
    cons(T_t) &= \{C_c^t\}_{c=1}^n \\
    args(C_c^t) &= \{T_j\}_{j=1}^m \\
    \lvert T_t \lvert &= \lvert cons(T_t) \lvert \ = n
  \end{align*}
  We will also define the \textit{branching factor from $C_i^u$ to $T_v$} as the
  natural number $\beta(T_v, C_i^u)$ denoting the number of occurrences of $T_v$
  in the arguments of $C_i^u$:
  \begin{align*}
    \beta(T_v, C_i^u) = \lvert \{ T_k \in args(C_i^u)\ \lvert \ T_k = T_v \} \lvert
  \end{align*}
\end{helpers-def}

Before showing our main theorem, we need some preliminary propositions.
The following one relates the mean of reproduction of constructors with their
types and the number of occurrences in the ADT declaration.

\begin{mC}
  Let $M_C$ be the mean matrix for constructors for a given, possibly mutually
  recursive data types $\{T_t\}_{t=1}^n$ and type constructors
  $\{C^t_i\}_{i=1}^{\lvert T_t \lvert}$.
  Assuming $p_{C_i^t}$ to be the probability of generating a constructor
  $C^t_i \in cons(T_t)$ whenever a value of type $T^t$ is needed, then it holds
  that:
  \begin{align}
    m_{C_i^u C_j^v} = \beta(T_v, C_i^u) \cdot p_{C_j^v}
    \label{eq:mC-prop}
  \end{align}
\end{mC}
\begin{proof}
  Let $m_{C_i^u C_j^v}$ be an element of $M_C$, we know that $m_{C_i^u C_j^v}$
  represents the expected number of constructors $C_j^v \in cons(T_v)$ generated
  whenever a constructor $C_i^u \in cons(T_u)$ is generated.
  Since every constructor is composed of a product of (possibly) many arguments,
  we need sum the expected number of constructors $C_j^v$ generated by each
  argument of $C_i^u$ of type $T_v$---the expected number of constructors
  $C_j^v$ generated by an argument of type different than $T_v$ is null.
  For this, we define the random variable $X_k^{C_i^u C_j^v}$ capturing the
  number of constructors $C_j^v$ generated by the $k$-th argument of $C_i^u$ as
  follows:
  \begin{align*}
    &X_k^{C_i^u C_j^v} : cons(T_v) \rightarrow \mathbb{N}\\
    &X_k^{C_i^u C_j^v} (C_c^v) =
      \begin{cases}
        1 & \mathrm{if}\  c = j\\
        0 & \mathrm{otherwise}
      \end{cases}
  \end{align*}
  We can calculate the probabilities of generating zero or one constructors
  $C_j^v$ by the $k$-th argument of $C_i^u$ as follows:
  \begin{align*}
    P(X_k^{C_i^u C_j^v} = 0) &= 1 - p_{C_j^v}\\
    P(X_k^{C_i^u C_j^v} = 1) &=  p_{C_j^v}
  \end{align*}
  Then, we can calculate the expectancy of each $X_k^{C_i^u C_j^v}$:
  \begin{align}
    E[X_k^{C_i^u C_j^v}]
    \ =\  1 \cdot P(X_k^{C_i^u C_j^v} = 1) + 0 \cdot P(X_k^{C_i^u C_j^v} = 0)
    \ =\  p_{C_j^v}
    \label{eq:exp1}
  \end{align}
  Finally, we can calculate the expected number of constructors $C_j^v$
  generated whenever we generate a constructor $C_i^u$ by adding the expected
  number of $C_j^v$ generated by each argument of $C_i^u$ of type $T_v$:
  \begin{align*}
    m_{C_i^u C_j^v}
    &= \sum_{\{T_k \in args(C_i^u)\ \lvert\  T_k = T_v\}} E[X_k^{C_i^u C_j^v}]\\
    &= \sum_{\{T_k \in args(C_i^u)\ \lvert\  T_k = T_v\}} p_{C_j^v}
    &(\text{by (\ref{eq:exp1})})\\
    &= p_{C_j^v} \cdot \sum_{\{T_k \in args(C_i^u)\ \lvert\  T_k = T_v\}} 1
    &(\text{$p_{C_j^v}$ is constant})\\
    &= p_{C_j^v} \cdot \lvert \{ T_k \in args(C_j^v)\ \lvert \ T_k = T_v \} \lvert
    &(\sum_S 1 = \lvert S \lvert)\\
    &=  p_{C_j^v} \cdot \beta(T_v, C_i^u)
    &(\text{by def. of $\beta$})
  \end{align*}
\end{proof}

The next propositions relates the mean of reproduction of types with their
constructors.
\begin{mT}
  Let $M_T$ be the mean matrix for types for a given, possibly mutually
  recursive data types $\{T_t\}_{t=1}^n$ and type constructors
  $\{C^t_i\}_{i=1}^{\lvert T_t \lvert}$.
  Assuming $p_{C_i^t}$ to be the probability of generating a constructor
  $C^t_i \in cons(T_t)$ whenever a value of type $T_t$ is needed, then it holds
  that:
\begin{align}
  m_{T_u T_v} = \sum_{C_k^u \in cons(T^u)} \beta(T_v, C_k^u) \cdot p_{C_k^u}
  \label{eq:mT-prop}
\end{align}
\end{mT}
\begin{proof}
  Let $m_{T_uT_v}$ be an element of $M_T$, we know that $m_{T_uT_v}$ represents
  the expected number of placeholders of type $T_v$ generated whenever a
  placeholder of type $T_u$ is generated, i.e. by any of its constructors.
  Therefore, we need to average the number of place holders of type $T_v$
  appearing on each constructor of $T_u$.
  For that, we introduce the random variable $Y^{uv}$ capturing this behavior.
%
%
  \begin{align*}
    Y^{uv} : cons(T_u) \rightarrow \mathbb{N}\\
    Y^{uv} (C_k^u) = \beta(T_v, C_k^u)
  \end{align*}
  And we can obtain $m_{T_u T_v}$ by calculating the expected value of $Y^{uv}$
  as follows.
%
%
  \begin{align*}
    m_{T_u T_v}
    &= E[Y^{uv}]\\
    &= \sum_{C_k^u \,\in\, cons(T_u)} \beta(T_v, C_k^u) \cdot P(Y^{uv} = C_k^u)
    &(\text{def. of $E[Y^{uv}]$})\\
    &= \sum_{C_k^u \,\in\, cons(T_u)} \beta(T_v, C_k^u) \cdot p_{C_k^u}
    &(\text{def. of $p_{C_k^u}$})
  \end{align*}
\end{proof}

The next proposition relates one entry in $M_T$ with its corresponding in $M_C$.
\begin{mC-vs-mT}
  Let $M_C$ and $M_T$ be the mean matrices for constructors and types
  respctively for a given, possibly mutually recursive data types
  $\{T_t\}_{t=1}^n$ and type constructors $\{C^t_i\}_{i=1}^{\lvert T_t \lvert}$.
  Assuming $p_{C_i^t}$ to be the probability of generating a type constructor
  $C^t_i \in cons(T_t)$ whenever a value of type $T_t$ is needed, then it holds
  that:
  \begin{align}
    p_{C_i^v} \cdot m_{T_u T_v} = \sum_{C_j^u \in cons(T_u)} m_{C_j^u C_i^v} \cdot p_{C_j^u}
    \label{eq:mC-vs-mT}
  \end{align}
\end{mC-vs-mT}
\begin{proof}
  Let $C_i^u$ and $C_j^v$ be type constructors of $T^u$ and $T^v$ respectively.
  Then, by (\ref{eq:mC-prop}) and (\ref{eq:mT-prop}) we have:
  \begin{align}
    m_{C_i^u C_j^v} &= \beta(T_v, C_i^u) \cdot p_{C_j^v}
    \qquad \label{eq:mC1} \\
    m_{T_u T_v} &= \sum_{C_k^u \in cons(T_u)} \beta(T_v, C_k^u) \cdot p_{C_k^u}
    \qquad \label{eq:mT1}
  \end{align}
  Now, we can rewrite (\ref{eq:mC1}) as follows:
  \begin{align}
    \beta(T_v, C_i^u) &= \frac{m_{C_i^u C_j^v}}{p_{C_j^v}}
    \qquad (\text{if}\  p_{C_j^v} \ne 0) \label{eq:mC2}
  \end{align}
(In the case that $p_{C_j^v} = 0$, the last equation in this proposition holds
trivially by (\ref{eq:mC1}).)
  And by replacing (\ref{eq:mC2}) in (\ref{eq:mT1}) we obtain:
  \begin{align*}
    m_{T_u T_v} &= \sum_{C_k^u \in cons(T_u)} \frac{m_{C_i^u C_j^v}}{p_{C_j^v}} \cdot p_{C_k^u}\\
    m_{T_u T_v} &= \frac{1}{p_{C_j^v}} \cdot \sum_{C_k^u \in cons(T_u)} m_{C_i^u C_j^v} \cdot p_{C_k^u}
    &(\text{$p_{C_j^v}$ constant})\\
    p_{C_j^v} \cdot m_{T_u T_v}  &= \sum_{C_k^u \in cons(T_u)} m_{C_i^u C_j^v} \cdot p_{C_k^u}
  \end{align*}
\end{proof}

\noindent
Now, we proceed to prove our main result.
\begin{types-matrix}
  Consider a QuickCheck generator for a (possibly) mutually recursive data types
  $\{T_t\}_{t=1}^k$ and type constructors $\{C_i^t\}_{i=1}^{\lvert T^t \lvert}$.
  We assume $p_{C^t_i}$ as the probability of generating a type constructor
  $C^t_i \in cons(T_t)$ when a value of type $T_t$ is needed.
  We will call $T_r$ $(1 \le r \le k)$ to the generation root data type, and
  $M_C$ and $M_T$ to the mean matrices for the multi-type branching process
  capturing the generation behavior of type constructors and types respectively.
%
%
  The branching process predicting the expected number of type constructors at
  level $n$ is governed by the formula:
  \begin{align*}
    E[G_n^C]^T = E[G_0^C]^T \cdot \left( \frac{I - (M_C)^{n+1}}{I - M_C} \right)
  \end{align*}
  In the same way, the branching process predicting the expected number of type
  placeholders at level $n$ is given by:
  \begin{align*}
    E[G_n^T]^T = E[G_0^T]^T \cdot \left( \frac{I - (M_T)^{n+1}}{I - M_T} \right)
  \end{align*}
  where $G_n^C$ denotes the constructors population at the level $n$, and
  $G_n^T$ denotes the type placeholders population at the level $n$.
  The expected number of constructors $C_i^t$ at the $n$-th level is given by
  the expected constructors population at the $n$-level $E[G_n^C]$ indexed by
  the corresponding constructor.
  Similarly, the expected number of placeholders of type $T_t$ at the $n$-th
  level is given by the expected types population at the $n$-level $E[G_n^T]$
  indexed by the corresponding type.
  The initial constructors population $E[G_0^C]$ is defined as the probability
  of each constructor if it belongs to the root data type, and zero if it belong
  to any other data type:
  \begin{align*}
    E[G_0^C].C_i^t  =
    \begin{cases}
      p_{C_i^t} & \mathrm{if}\ t = r\\
      0 & \mathrm{otherwise}
    \end{cases}
  \end{align*}
  The initial type placeholders population is defined as the almost surely
  probability for the root type, and zero for any other type:
  \begin{align*}
    E[G_0^T].T_t =
    \begin{cases}
      1 & \mathrm{if}\  t = r\\
      0 & \mathrm{otherwise}
    \end{cases}
  \end{align*}
  Finally, it holds that:
  \begin{align*}
    (E[G_n^C]).C_i^t = (E[G_n^T]).T^t \cdot p_{C_i^t}
  \end{align*}
  In other words, the expected number of constructors $C^t_i$ at the $n$-th
  level consists of the expected number of placeholders of its type (i.e.,
  $T_t$) at level $n$ times the probability to generate that constructor.
\end{types-matrix}
\begin{proof}
  By induction on the generation size $n$.\\
  \begin{itemize}
  \item \textbf{Base case}\\
    We want to prove $(E[G_0^C]).C_i^t = (E[G_0^T]).T_t \cdot p_{C_i^t}$.\\
    Let $T_t$ be a data type from the Galton-Watson branching process.\\
  \begin{itemize}
  \item If $T_t = T_r$ then by the definitions of the initial type constructors
    and type placeholders populations we have:
    \begin{align*}
      (E[C_0^C]).C_i^t = p_{C_i^t} \qquad\qquad (E[G_0^T]).T_t = 1
    \end{align*}
    And the theorem trivially holds by replacing $(E[G_0^C]).C_i^t$ and
    $(E[G_0^T]).T_t$ with the previous equations in the goal.\\
  \item If $T_t \neq T_r$ then by the definitions of the initial type
    constructors and type placeholders populations we have:
    \begin{align*}
      (E[G_0^C]).C_i^t = 0 \qquad\qquad (E[G_0^T]).T_t = 0
    \end{align*}
    And once again, the theorem trivially holds by replacing $(E[G_0^C]).C_i^t$
    and ($E[G_0^T]).T_t$ with the previous equations in the goal.
  \end{itemize}
  \vspace{10pt}
\newpage
\nobalance
  \item \textbf{Inductive case}\\
    We want to prove $(E[G_n^C]).C_i^t = (E[G_n^T]).T_t \cdot p_{C_i^t}$.\\
    For simplicity, we will call $\Gamma = \{T_t\}_{t=1}^k$.
  \begin{align*}
    (E[G_n^C]).C_i^t
    &= E \left[ \sum_{T_k \in \Gamma} \left( \sum_{C_j^k \in cons(T_k)} (G_{(n-1)}^C).C_j^k \cdot m_{C_j^k C_i^t} \right) \right]
    &(\text{by G.W. proc.})\\
    &= \sum_{T_k \in \Gamma} E \left[ \sum_{C_j^k \in cons(T_k)} (G_{(n-1)}^C).C_j^k \cdot m_{C_j^k C_i^t}  \right]
    &(\text{by prob.})\\
    &= \sum_{T_k \in \Gamma} \left( \sum_{C_j^k \in cons(T_k)} E [ (G_{(n-1)}^C).C_j^k \cdot m_{C_j^k C_i^t} ] \right)
    &(\text{by prob.})\\
    &= \sum_{T_k \in \Gamma}\left( \sum_{C_j^k \in cons(T_k)} E [(G_{(n-1)}^C).C_j^k ] \cdot m_{C_j^k C_i^t} \right)
    &(\text{by prob.})\\
    &= \sum_{T_k \in \Gamma}\left( \sum_{C_j^k \in cons(T_k)} (E [G_{(n-1)}^C]).C_j^k \cdot m_{C_j^k C_i^t} \right)
    &(\text{by linear alg.})\\
    &= \sum_{T_k \in \Gamma}\left( \sum_{C_j^k \in cons(T_k)} (E[G_{(n-1)}^T]).T_t \cdot p_{C_j^k} \cdot m_{C_j^k C_i^t} \right)
    &(\text{by I.H.})\\
    &= \sum_{T_k \in \Gamma} (E[G_{(n-1)}^T]).T_t \cdot \sum_{C_j^k \in cons(T_k)}  p_{C_j^k} \cdot m_{C_j^k C_i^t}
    &(\text{by linear alg.})\\
    &= \sum_{T_k \in \Gamma} (E[G_{(n-1)}^T]).T_t \cdot p_{C_i^t} \cdot m_{T_k T_t}
    &(\text{by (\ref{eq:mC-vs-mT})})\\
    &= \sum_{T_k \in \Gamma} (E[G_{(n-1)}^T]).T_t \cdot m_{T_k T_t} \cdot p_{C_i^t}
    &(\text{rearrange})\\
    &= \sum_{T_k \in \Gamma} E[(G_{(n-1)}^T).T_t] \cdot m_{T_k T_t} \cdot p_{C_i^k}
    &(\text{by linear alg.})\\
    &= \sum_{T_k \in \Gamma} E[(G_{(n-1)}^T).T_t \cdot m_{T_k T_t}] \cdot p_{C_i^t}
    &(\text{by prob.})\\
    &= E \left[ \sum_{T_k \in \Gamma} (G_{(n-1)}^T).T_t \cdot m_{T_k T_t} \right] \cdot p_{C_i^t}
    &(\text{by prob.})\\
    &= (E [ G_n^T]).T_t \cdot p_{C_i^t}
    &(\text{by G.W. proc.})
  \end{align*}
\end{itemize}
\end{proof}
%
\clearpage
\section{Additional information}
\label{sec:appendixB}

\balance

This appendix is meant to provide further analyses for the aspects presented
throughout this work that would not fit into the available space.
%

\subsection{Termination issues with library \derive}
\label{app:derive}

\begin{figure}[b]
  \includegraphics[width=0.75\columnwidth]{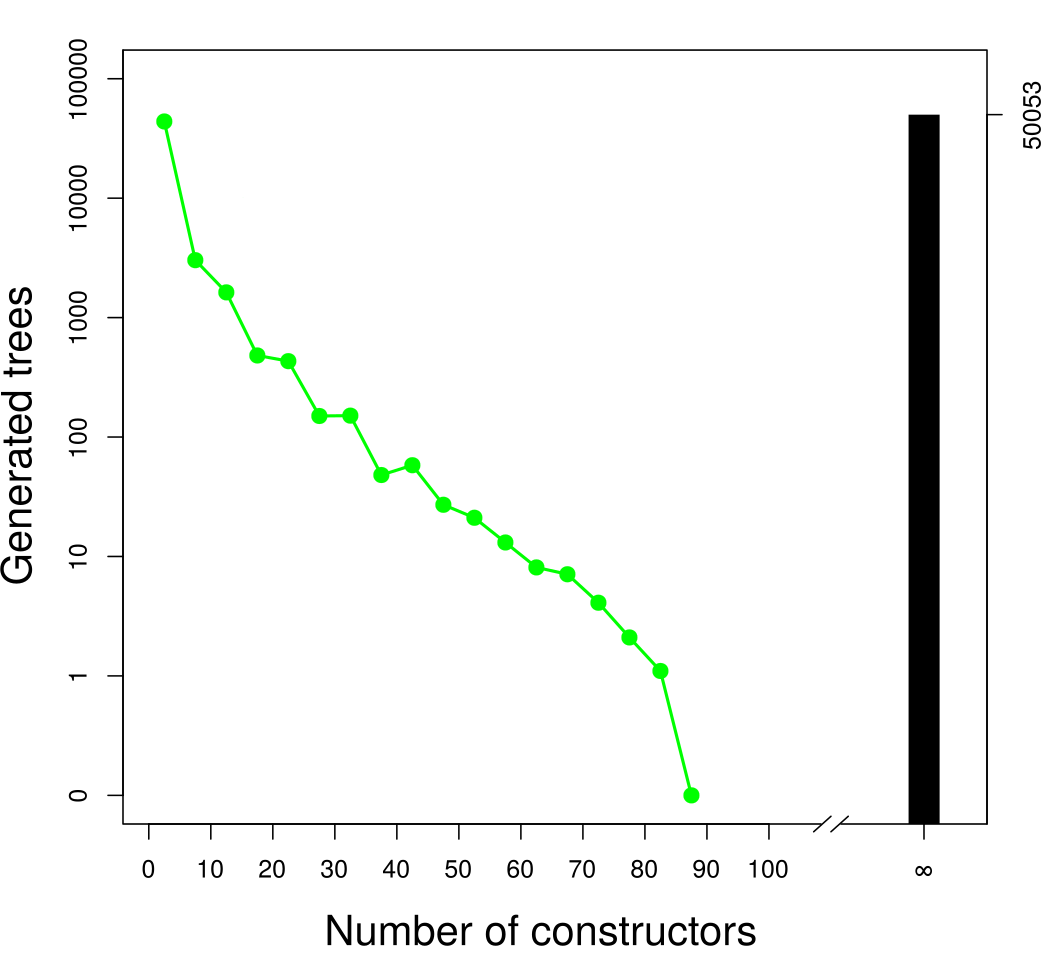}
  \caption{
    \label{fig:deriveloops} Distribution of (the amount of) \ensuremath{\Conid{T}} constructors
    induced by \derive.}
\end{figure}

As we have introduced in Section \ref{sec:QC}, the library \derive provides an
easy alternative to automatically synthesize random generators in compile time.
However, in presence of recursive data types, the generators obtained with this
tool lack mechanisms to ensure termination.
For instance, consider the following data type definition and its corresponding
generator obtained with \derive:
\begin{framed}
\begin{hscode}\SaveRestoreHook
\column{B}{@{}>{\hspre}l<{\hspost}@{}}%
\column{3}{@{}>{\hspre}l<{\hspost}@{}}%
\column{5}{@{}>{\hspre}l<{\hspost}@{}}%
\column{24}{@{}>{\hspre}c<{\hspost}@{}}%
\column{24E}{@{}l@{}}%
\column{27}{@{}>{\hspre}l<{\hspost}@{}}%
\column{E}{@{}>{\hspre}l<{\hspost}@{}}%
\>[3]{}\mathbf{data}\;\Conid{T}\mathrel{=}\Conid{A}\mid \Conid{B}\;\Conid{T}\;\Conid{T}\mid \Conid{C}\;\Conid{T}\;\Conid{T}{}\<[E]%
\\[\blanklineskip]%
\>[3]{}\mathbf{instance}\;\Conid{Arbitrary}\;\Conid{T}\;\mathbf{where}{}\<[E]%
\\
\>[3]{}\hsindent{2}{}\<[5]%
\>[5]{}\Varid{arbitrary}\mathrel{=}\Varid{oneof}\;{}\<[24]%
\>[24]{}[\mskip1.5mu {}\<[24E]%
\>[27]{}\Varid{pure}\;\Conid{A}{}\<[E]%
\\
\>[24]{},{}\<[24E]%
\>[27]{}\Conid{B}\mathop{\langle \texttt{\$} \rangle}\Varid{arbitrary}\mathop{\langle \ast \rangle}\Varid{arbitrary}{}\<[E]%
\\
\>[24]{},{}\<[24E]%
\>[27]{}\Conid{C}\mathop{\langle \texttt{\$} \rangle}\Varid{arbitrary}\mathop{\langle \ast \rangle}\Varid{arbitrary}\mskip1.5mu]{}\<[E]%
\ColumnHook
\end{hscode}\resethooks
\end{framed}
When using this generator, \emph{every constructor in the obtained generator has
  the same probability of being chosen}.
Aditionally, at each point of the generation process, if we randomly generate a
recursive type constructor (either \ensuremath{\Conid{B}} or \ensuremath{\Conid{C}}), then we also need to generate
two new \ensuremath{\Conid{T}} values in order to fill the arguments of the chosen type
constructor.
As a result, it is expected (on average) that each time \quickcheck generates a
recursive constructor (i.e., \ensuremath{\Conid{B}} or \ensuremath{\Conid{C}}) at one level, {\em more than one
  recursive constructor is generated at the next level}---thus, frequently
leading to an infinite generation loop.
This behavior can be formalized using the concept known as \emph{probability
  generating function}, where it is proven that the extinction probability of a
generated value $d$ (and thus the termination of the generation) can be
calculated by finding the smallest fix point of the generation recurrence.
In our example, this is the smallest $d$ such that $d = P_A + (P_B + P_C) \cdot
d^2 = (1/3) + (2/3) \cdot d^2$, where $P_i$ denotes the probabilty of generating
a $i$ constructor.
In this case $d = 1/2$.

Figure \ref{fig:deriveloops} provides an empirical verification of this
non-terminating behavior.
It shows the distribution (in terms of amount of constructors) of 100000
randomly generated \ensuremath{\Conid{T}} values obtained using the \derive generator shown above.
The black bar on the right represents the amount of values that induced an
infinite generation loop.
Such values were recognized using a sufficiently big timeout.
The random generation gets stuck in an infinite generation loop almost exactly
half of the times we generate a random \ensuremath{\Conid{T}} value.

In practice, this non terminating behavior gets worse as we increase either the
number of recursive constructors or the number of their recursive arguments in
the data type definition, since this increases the probability of choosing a
recursive constructor each time we need to generate a subterm.

\subsection{Multi-type Branching Processes}
\label{app:expectation}

We will verify the soundness of the step noted as $(\star)$, used to deduce
$E[G_n^j | G_{n-1}]$ in Section \ref{sec:bp2}.
In first place, note that $E[G_n^j | G_{n-1}]$ can be rewritten as:
\begin{align*}
  E[G_n^j | G_{n-1}] = E \left[ \sum_{i=1}^{d} \sum_{p=1}^{G_{n-1}} \xi_{ij}^p \right]
\end{align*}
Where symbol $\xi_{ij}^p$ denotes the number of offspring of kind $j$ that the
parent $p$ of kind $i$ produces.
If the parent $p$ has not kind $i$, then $\xi_{ij}^p = 0$.
Essentially, the sums simply iterate on all of the different kind of parents
present in the $n$th-generation, counting the number of offspring of kind $j$
that they produce.
Then, since the expectation of the sum is the sum of expectation, we have that:
\begin{align*}
E[G_n^j | G_{n-1}] = \sum_{i=1}^{d} \sum_{p=1}^{G_{n-1}} E \left[ \xi_{ij}^p \right]
\end{align*}
In the inner sum, there are some terms which are $0$ and others which are the
expected offspring of kind $j$ that a parent of kind $i$ produces.
As introduced in Section \ref{sec:bp2}, we capture with random variable $R_{ij}$
the distribution governing that a parent of kind $i$ produces offspring of kind
$j$.
Finally, by filtering out all the terms which are $0$ in the inner sum, i.e.,
where $p \neq i$, we obtain the expected result:
\begin{align*}
E[G_n^j | G_{n-1}] = \sum_{i=1}^{d} G_{(n-1)}^i\!\cdot\! E [R_{ij}]
\end{align*}

\subsection{Terminal constructors}
\label{app:terminals}

As we explained in Section \ref{sec:terminals}, our tool sinthesizes random
generators for which the generation of terminal constructors can be thought of
two different random processes.
More specifically, the first ($n-1$) generations of the branching process are
composed of a mix of non-terminals and terminals constructors.
The last level, however, only contains terminal constructors since the size
limit has been reached.
Figure \ref{fig:terminals} shows a graphical representation of the overall
process.

\begin{figure}[H] 
\newcommand{\term}{{\color{red} \small $\blacksquare$}}
\newcommand{\nonterm}{{\color{blue} \huge $\bullet$}}
\centering
\begin{tikzpicture}
  [ level 1+/.style={level distance=1cm, sibling distance = 0.25cm} ]
    \Tree
      [.\node(1){\nonterm};
        [.\node(2){\nonterm};
        \edge[densely dotted, thick] node[auto=left]{};
        [.\node(4){\nonterm};
          \node(8){\term};
          \node(9){\term};
        ]
        \edge[densely dotted, thick] node[auto=left]{};
        [.\node(5){\nonterm};
           \node(10){\term};
           \node(11){\term};
        ]
      ]
      [.\node(16){\term};
      ]
      [.\node(3){\nonterm};
          \node(6){\term};
          \edge[densely dotted, thick] node[auto=left]{};
          [.\node(7){\nonterm};
             \node(14){\term};
             \node(15){\term};
          ]
       ]
    ]
    \draw[densely dotted, thick](2)--(16);
    \draw[densely dotted, thick](16)--(3);
    \draw[densely dotted, thick](4)--(5);
    \draw[densely dotted, thick](6)--(7);
    \draw[densely dotted, thick](8)--(9);
    \draw[densely dotted, thick](10)--(11);
    \draw[densely dotted, thick](14)--(15);
    \node[right = 2.75cm of 1]          (Z0Text)  {$G_0$};
    \node[below = 0.45cm of Z0Text]     (Z1Text)  {$G_1$};
    \node[below = 0.45cm of Z1Text]     (Zn1Text) {$G_{n-1}$};
    \node[below = 0.50cm of Zn1Text]    (ZnText)  {$G_{n}$};
    \draw[densely dotted, thick](Z1Text)--(Zn1Text);
    \draw[dashed, gray] (-3.5,-2.5)--(4,-2.5);
\end{tikzpicture}
\caption{\label{fig:terminals} Generation processes of non-terminal
  ({\color{blue} \Large $\bullet$}) and terminal ({\color{red} \tiny
    $\blacksquare$}) constructors.}
\end{figure}
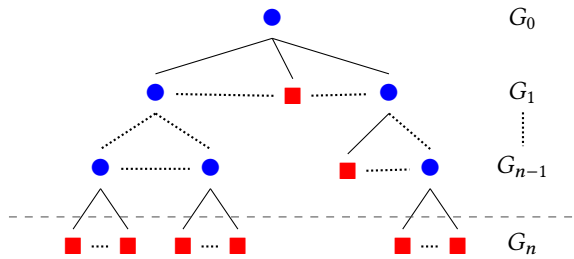

\subsection{Implementation}
\label{app:implementation}

In this subsection, will give more details on the implementation of our tool.
Firstly, Figure \ref{fig:pipeline} shows a schema for the automatic derivation
pipeline our tool performs.
The user provides a target data type, a cost function and a desired generation
size, and our tool returns an optimized random generator.
The components marked in red are heavily dependent on Template Haskell and refer
to the type introspection and code generation stages of \dragen, while the
intermediate stages (in blue) are composed by our prediction mechanishm and the
probabilities optimizator.
\begin{figure}[H]
  \includegraphics[width=\columnwidth]{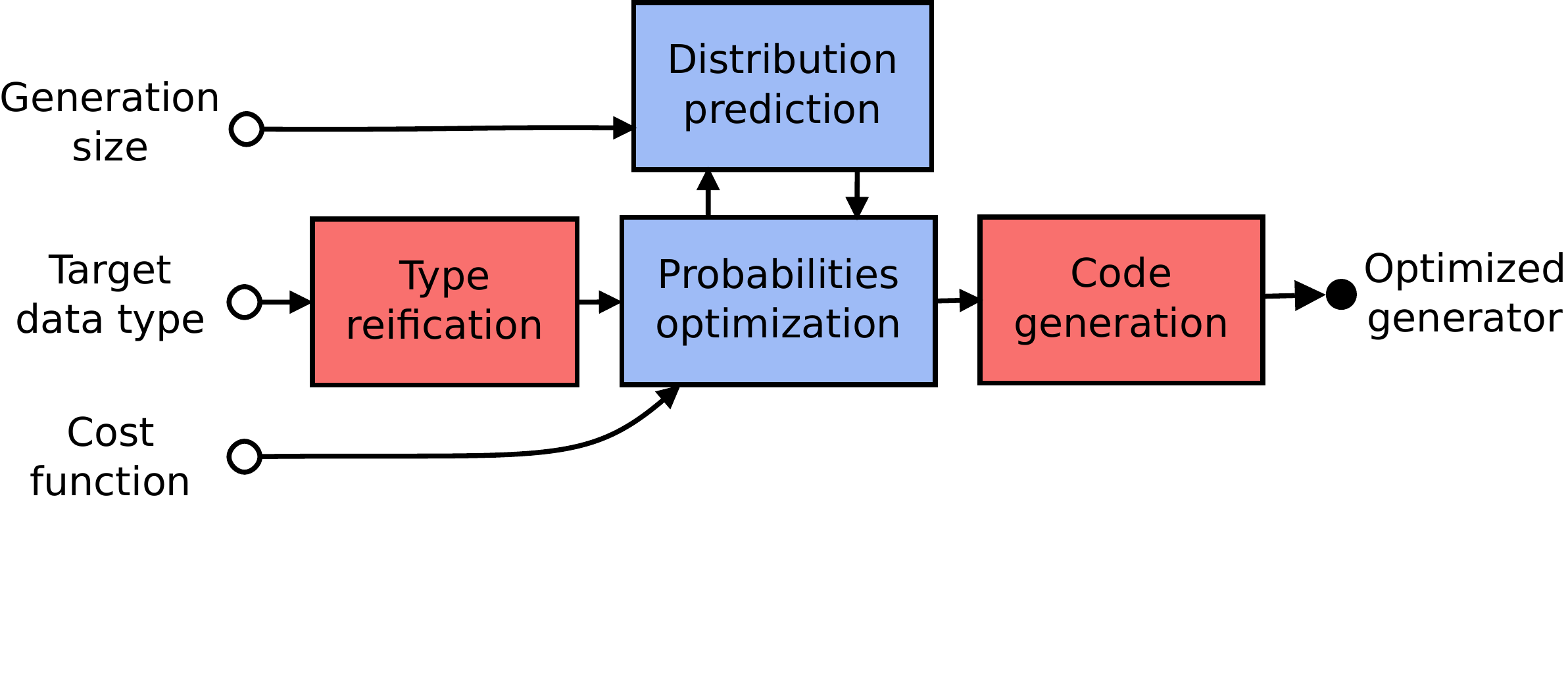}
  \vspace{-50pt}
  \caption{\label{fig:pipeline} Generation schema.}
\end{figure}

\paragraph{Cost functions}

The probabilities optimizer that our tool implements essentially works
minimizing a provided cost function that encodes the desired distribution of
constructors at the optimized generator.
As shown in Section \ref{sec:implementation}, \dragen comes with a minimal set
of useful cost functions.
Such functions are built around the \emph{Chi-Square Goodness of Fit Test}
\cite{chisquarebook}, a statistical test used quantify how the observed value of
a given phenomena is significantly different from its expected value:
\begin{align*}
  \chi^2 = \sum_{C_i \in\, \Gamma} \frac{(observed_i - expected_i) ^ 2}{expected_i}
\end{align*}
Where $\Gamma$ is a subset of the constructors involved in the generation
process; $observed_i$ corresponds to the predicted number of generated $C_i$
constructors; and $expected_i$ corresponds to the amount of constructors $C_i$
desired in the distribution of the optimized generator.
This fitness test was chosen for empirical reasons, since it provides better
results in practice when finding probabilities that ensure certain
distributions.

In this appendix we will take special attention to the \ensuremath{\Varid{weighted}} cost function,
since it is the most general one that our tool provides---the remaining cost
funcions provided could be expressed in terms of \ensuremath{\Varid{weighted}}.
This function uses our previously discussed prediction mechanism to obtain a
prediction of the constructors distribution under the current given
probabilities and the generation size (see \ensuremath{\Varid{obs}}), and uses it to calculate the
Chi-Square Goodness of Fit Test.
A simplified implementation of this cost function is as follows.

\begin{framed}
\begin{hscode}\SaveRestoreHook
\column{B}{@{}>{\hspre}l<{\hspost}@{}}%
\column{3}{@{}>{\hspre}l<{\hspost}@{}}%
\column{5}{@{}>{\hspre}l<{\hspost}@{}}%
\column{7}{@{}>{\hspre}l<{\hspost}@{}}%
\column{E}{@{}>{\hspre}l<{\hspost}@{}}%
\>[B]{}\Varid{weighted}\mathbin{::}[\mskip1.5mu (\Conid{Name},\Conid{Double})\mskip1.5mu]\to \Conid{CostFunction}{}\<[E]%
\\
\>[B]{}\Varid{weighted}\;\Varid{weights}\;\Varid{size}\;\Varid{probs}\mathrel{=}\Varid{chiSquare}\;\Varid{obs}\;\Varid{exp}{}\<[E]%
\\
\>[B]{}\hsindent{3}{}\<[3]%
\>[3]{}\mathbf{where}{}\<[E]%
\\
\>[3]{}\hsindent{2}{}\<[5]%
\>[5]{}\Varid{chiSquare}\mathrel{=}\Varid{sum}\;\!\circ\!\;\Varid{zipWith}\;(\lambda \Varid{o}\;\Varid{e}\to (\Varid{o}\mathbin{-}\Varid{e})\;\!^2\mathbin{/}\Varid{e}){}\<[E]%
\\
\>[3]{}\hsindent{2}{}\<[5]%
\>[5]{}\Varid{obs}\mathrel{=}\Varid{predict}\;\Varid{size}\;\Varid{probs}{}\<[E]%
\\
\>[3]{}\hsindent{2}{}\<[5]%
\>[5]{}\Varid{exp}\mathrel{=}\Varid{map}\;\Varid{weight}\;(\Varid{\Conid{Map}.keys}\;\Varid{probs}){}\<[E]%
\\
\>[3]{}\hsindent{2}{}\<[5]%
\>[5]{}\Varid{weight}\;\Varid{con}\mathrel{=}\mathbf{case}\;\Varid{lookup}\;\Varid{con}\;\Varid{weights}\;\mathbf{of}{}\<[E]%
\\
\>[5]{}\hsindent{2}{}\<[7]%
\>[7]{}\Conid{Just}\;\Varid{w}\to \Varid{w}\mathbin{*}\Varid{size}{}\<[E]%
\\
\>[5]{}\hsindent{2}{}\<[7]%
\>[7]{}\Conid{Nothing}\to \mathrm{0}{}\<[E]%
\ColumnHook
\end{hscode}\resethooks
\end{framed}

Note how we multiply each weight by the generation size provided by the user
(case \ensuremath{\Conid{Just}\;\Varid{w}}), as a simple way to control the relative size of the generated
values.
Moreover, the generation probabilities for the constructors not listed in the
proportions list do not contribute to the cost (case \ensuremath{\Conid{Nothing}}), and thus they
can be freely adjusted by the optimizer to fit the proportions of the listed
constructors.
In this light, the \ensuremath{\Varid{uniform}} cost function can be seen as a special case of
\ensuremath{\Varid{weighted}}, where every constructor is listed with weight 1.

\paragraph{Optimization algorithm}

As introduced in Section \ref{sec:implementation}, our tool makes use of an
optimization mechanishm in order to obtain a suitable generation probabilities
assignment for its derived generators.
Figure \ref{algo:optimize} illustrates a simplified implementation of our
optimization algorithm.
This optimizer works selecting recursively the most suitable neighbor, i.e., a
probability assignment that it close to the current one and that minimizes the
output of the provided cost function.
This process is repeated until a local minimum is found, when the are no further
neighbors that remains unvisited; or if the step improvement is below a minimum
predetermined $\varepsilon$.

In our setting, neighbors are obtained by taking the current probability
distribution, and constructing a list of paired probability distributions, where
each one is constructed from the current distribution, adjusting each
constructor probability by $\pm\Delta$.
This behavior is shown in Figure \ref{algo:neighborhood}.
Note the need of bound checking and normalization of the new neighbors in order
to enforce a probability distribution (\ensuremath{\Varid{max}\;\mathrm{0}} and \ensuremath{\Varid{norm}}).
Each pair of neighbors is then joined together and returned as the current
probability distribution immediate neighborhood.
\begin{figure}[H]
\begin{framed}
\begin{hscode}\SaveRestoreHook
\column{B}{@{}>{\hspre}l<{\hspost}@{}}%
\column{3}{@{}>{\hspre}l<{\hspost}@{}}%
\column{5}{@{}>{\hspre}l<{\hspost}@{}}%
\column{7}{@{}>{\hspre}l<{\hspost}@{}}%
\column{24}{@{}>{\hspre}l<{\hspost}@{}}%
\column{E}{@{}>{\hspre}l<{\hspost}@{}}%
\>[B]{}\Varid{optimize}\mathbin{::}\Conid{CostFunction}\to \Conid{Size}\to \Conid{ProbMap}\to \Conid{ProbMap}{}\<[E]%
\\
\>[B]{}\Varid{optimize}\;\Varid{cost}\;\Varid{size}\;\Varid{init}\mathrel{=}\Varid{localSearch}\;\Varid{init}\;[\mskip1.5mu \mskip1.5mu]\;\mathbf{where}{}\<[E]%
\\
\>[B]{}\hsindent{3}{}\<[3]%
\>[3]{}\Varid{localSearch}\;\Varid{focus}\;\Varid{visited}{}\<[E]%
\\
\>[3]{}\hsindent{2}{}\<[5]%
\>[5]{}\mid \Varid{null}\;\Varid{new}{}\<[24]%
\>[24]{}\mathrel{=}\Varid{focus}{}\<[E]%
\\
\>[3]{}\hsindent{2}{}\<[5]%
\>[5]{}\mid \Varid{gain}\leq \varepsilon{}\<[24]%
\>[24]{}\mathrel{=}\Varid{focus}{}\<[E]%
\\
\>[3]{}\hsindent{2}{}\<[5]%
\>[5]{}\mid \Varid{otherwise}{}\<[24]%
\>[24]{}\mathrel{=}\Varid{localSearch}\;\Varid{best}\;\Varid{frontier}{}\<[E]%
\\
\>[3]{}\hsindent{2}{}\<[5]%
\>[5]{}\mathbf{where}{}\<[E]%
\\
\>[5]{}\hsindent{2}{}\<[7]%
\>[7]{}\Varid{best}\mathrel{=}\Varid{minimumBy}\;(\Varid{comparing}\;(\Varid{cost}\;\Varid{size}))\;\Varid{new}{}\<[E]%
\\
\>[5]{}\hsindent{2}{}\<[7]%
\>[7]{}\Varid{new}\mathrel{=}\Varid{neighbors}\;\Varid{focus}\mathbin{\char92 \char92 }(\Varid{focus}\mathbin{:}\Varid{visited}){}\<[E]%
\\
\>[5]{}\hsindent{2}{}\<[7]%
\>[7]{}\Varid{frontier}\mathrel{=}\Varid{new}\plus \Varid{visited}{}\<[E]%
\\
\>[5]{}\hsindent{2}{}\<[7]%
\>[7]{}\Varid{gain}\mathrel{=}\Varid{cost}\;\Varid{size}\;\Varid{focus}\mathbin{-}\Varid{cost}\;\Varid{size}\;\Varid{best}{}\<[E]%
\ColumnHook
\end{hscode}\resethooks
\end{framed}
\caption{Optimization algorithm.}
\label{algo:optimize}
\end{figure}
\vspace{-10pt}

\begin{figure}[H]
\begin{framed}
\begin{hscode}\SaveRestoreHook
\column{B}{@{}>{\hspre}l<{\hspost}@{}}%
\column{3}{@{}>{\hspre}l<{\hspost}@{}}%
\column{10}{@{}>{\hspre}l<{\hspost}@{}}%
\column{25}{@{}>{\hspre}l<{\hspost}@{}}%
\column{53}{@{}>{\hspre}l<{\hspost}@{}}%
\column{64}{@{}>{\hspre}l<{\hspost}@{}}%
\column{E}{@{}>{\hspre}l<{\hspost}@{}}%
\>[B]{}\Varid{neighbors}\mathbin{::}\Conid{ProbMap}\to [\mskip1.5mu \Conid{ProbMap}\mskip1.5mu]{}\<[E]%
\\
\>[B]{}\Varid{neighbors}\;\Varid{probs}\mathrel{=}\Varid{concatMap}\;\Varid{perturb}\;(\Varid{\Conid{Map}.keys}\;\Varid{probs}){}\<[E]%
\\
\>[B]{}\hsindent{3}{}\<[3]%
\>[3]{}\mathbf{where}\;{}\<[10]%
\>[10]{}\Varid{perturb}\;\Varid{con}\mathrel{=}{}\<[25]%
\>[25]{}[\mskip1.5mu \Varid{norm}\;(\Varid{adj}\;{}\<[53]%
\>[53]{}(\mathbin{+}\Delta)\;{}\<[64]%
\>[64]{}\Varid{con}){}\<[E]%
\\
\>[25]{},\Varid{norm}\;(\Varid{adj}\;(\Varid{max}\;\mathrm{0}\;\!\circ\!\;{}\<[53]%
\>[53]{}(\mathbin{-}\Delta))\;{}\<[64]%
\>[64]{}\Varid{con})\mskip1.5mu]{}\<[E]%
\\
\>[10]{}\Varid{norm}\;\Varid{m}\mathrel{=}\Varid{fmap}\;(\mathbin{/}\Varid{sum}\;(\Varid{\Conid{Map}.elems}\;\Varid{m}))\;\Varid{m}{}\<[E]%
\\
\>[10]{}\Varid{adj}\;\Varid{f}\;\Varid{con}\mathrel{=}\Varid{\Conid{Map}.adjust}\;\Varid{f}\;\Varid{con}\;\Varid{probs}{}\<[E]%
\ColumnHook
\end{hscode}\resethooks
\end{framed}
\caption{Immediate neighbors of a probability distribution.}
\label{algo:neighborhood}
\end{figure}

\subsection{Case studies}
\label{app:casestudies}

\begin{table}[b]
  \begin{center}
    \caption{\label{tab:casestudies:adts} Type information for ADTs used in the
      case studies.}
    \begin{tabular}{l c c c c}
      \toprule 
      \thead{Case\\ Study}
      & \thead{Number of\\ involved\\ types}
      & \thead{Number of\\ involved\\ constructors}
      & \thead{Composite\\ types}
      & \thead{Mutually\\ recursive\\ types}\\
      \midrule 
      Lisp \hfill
      & 7 & 14 & Yes & Yes \\
      Bash \hfill
      & 31 & 136 & Yes & Yes \\
      Gif  \hfill
      & 16 & 30 & Yes & No \\
      \bottomrule 
    \end{tabular}
  \end{center}
  \vspace{10pt}
\end{table}
As explained in Section \ref{sec:casestudies}, our test cases targeted three
complex programs to evaluate the power of our derivation tool, i.e. \emph{GNU
  CLISP 2.49}, \emph{GNU bash 4.4} and \emph{GIFLIB 5.1}.
We derived random generators for each test case input format using some existent
Haskell libraries.
Each one of these libraries contains data types definition encoding the
structure of the input format of its corresponding test case, as well as
serialization functions that we use to convert randomly generated Haskell values
into actual test input files.
Table \ref{tab:casestudies:adts} illustrates the complexity of the bridging
libraries used in our case studies.

\paragraph{Testing runtimes}

As we have shown, \megadeth tends to derive generators which produce very small
test cases.
However, in our tests, the size differences in the test cases generated by each
tool does not produce remarkable differences in the runtimes required to test
each corpora.
Figure \ref{fig:runtimes} shows the execution time required to test each case of
the biggest corpora previously generated by each tool consisting of 1000 test
cases.
\begin{figure}[H]
  \centering
  \vspace{10pt}
  \includegraphics[width=\columnwidth]{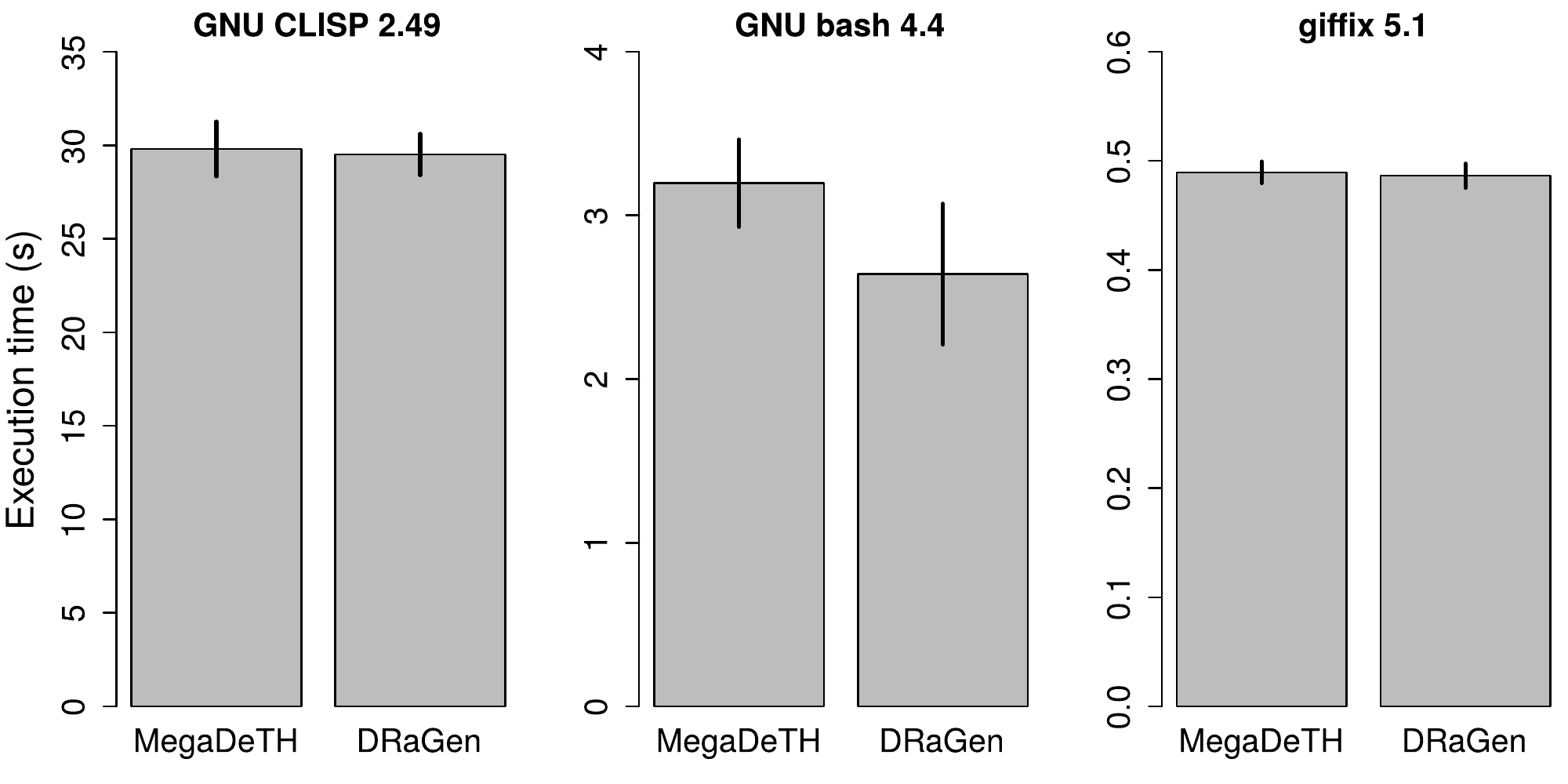}
  \caption{Execution time required to test the biggest randomly generated
    corpora consisting of 1000 files.}
  \label{fig:runtimes}
\end{figure}
}

\end{document}